\def\draft{1}
\def\doubleblind{0}
\documentclass[11pt,letterpaper]{article}

\usepackage[margin=1in]{geometry}
\usepackage[utf8]{inputenc}
\usepackage{amsthm}
\usepackage{amsthm, amsmath, amssymb, mathtools}
\usepackage{bm,bbm}
\usepackage{graphicx}
\usepackage{color,xcolor}
\usepackage{paralist,enumitem}
\usepackage[normalem]{ulem}
\usepackage{tabularx}
\usepackage{tikz}
\usepackage{caption,comment}
\usepackage{algorithmicx}
\usepackage{algorithm}
\usepackage{algpseudocode}
\usepackage{lipsum}
\newcounter{algsubstate}
\renewcommand{\thealgsubstate}{\alph{algsubstate}}

\algnewcommand\algorithmicinput{\textbf{Input:}}
\algnewcommand\Input{\item[\algorithmicinput]}

\algnewcommand\algorithmicoutput{\textbf{Output:}}
\algnewcommand\Output{\item[\algorithmicoutput]}

\algnewcommand\algorithmicgoal{\textbf{Goal:}}
\algnewcommand\Goal{\item[\algorithmicgoal]}

\newcommand{\blind}[2]{{\ifnum\draft=1\color{purple}\fi \ifnum\doubleblind=1#2\fi\ifnum\doubleblind=0#1\fi\ifnum\doubleblind=2$\{$ #1 $\vert$ #2 $\}$\fi}}

\makeatletter
\newcommand{\algmargin}{\the\ALG@thistlm}
\algnewcommand{\parState}[1]{\State%
  \parbox[t]{\dimexpr\linewidth-\algmargin}{\strut #1\strut}}
\usepackage[colorlinks=true, allcolors=blue]{hyperref}
\usepackage[capitalize,noabbrev,nameinlink]{cleveref}
\usepackage{booktabs}
\usepackage{subcaption}

\usepackage{silence}
\usepackage{xspace}

\usepackage[colorinlistoftodos,textsize=tiny,textwidth=2cm,color=green!50!gray]{todonotes}

\ifdefined\DEBUG
\newcommand{\amatodo}[1]{\textcolor{orange}{#1}}% inline comment, disappears
\newcommand{\amati}[1]{\textcolor{orange}{[AS:#1]}}% inline comment, disappears
\newcommand{\amatr}[1]{\todo[color=orange!100!black!20]{AS: #1}}%side comment, disappears

\newcommand{\snote}[1]{\textcolor{red}{[SV: #1]}}% inline comment, disappears
\else

\newcommand{\amati}[1]{}
\newcommand{\amatr}[1]{}
\newcommand{\amatodo}[1]{}% inline comment, disappears

\newcommand{\snote}[1]{}
\newcommand{\santi}[1]{}
\fi

\usepackage[style=alphabetic, backend=biber, minalphanames=3, maxalphanames=4, maxbibnames=99, maxcitenames=4]{biblatex}

\usepackage{amsthm}
\usepackage{thmtools,thm-restate}

\numberwithin{equation}{section}
\declaretheoremstyle[bodyfont=\it,qed=\qedsymbol]{noproofstyle}

\declaretheorem[name=Observation,numbered=no]{observation*}

\declaretheorem[numberlike=equation]{theorem}

\declaretheorem[name=Theorem,numbered=no]{theorem*}

\declaretheorem[numberlike=equation]{lemma}
\declaretheorem[name=Lemma,numbered=no]{lemma*}

\declaretheorem[numberlike=equation]{corollary}
\declaretheorem[name=Corollary,numbered=no]{corollary*}

\declaretheorem[name=Proposition,numbered=no]{proposition*}

\declaretheorem[name=Claim,numbered=no]{claim*}

\declaretheorem[name=Conjecture,numbered=no]{conjecture*}

\declaretheorem[name=Question,numbered=no]{question*}

\declaretheoremstyle[spaceabove=12pt,spacebelow=12pt,bodyfont=\it]
{defstyle} 

\declaretheorem[numberlike=equation,style=defstyle]{definition}
\declaretheorem[unnumbered,name=Definition,style=defstyle]{definition*}

\declaretheorem[numberlike=equation,style=defstyle]{example}
\declaretheorem[unnumbered,name=Example,style=defstyle]{example*}

\declaretheorem[unnumbered,name=Notation=defstyle]{notation*}

\declaretheorem[unnumbered,name=Construction,style=defstyle]{construction*}

\declaretheoremstyle[]{rmkstyle}

\crefname{claim}{claim}{claims}
\crefname{fact}{fact}{facts}
\crefname{corollary}{corollary}{corollaries}
\crefname{theorem}{theorem}{theorems}
\crefname{table}{table}{tables}

\Crefname{theorem}{Theorem}{Theorems}
\Crefname{claim}{Claim}{Claims}
\Crefname{fact}{Fact}{Facts}
\Crefname{corollary}{Corollary}{Corollaries}
\Crefname{table}{Table}{Tables}

\newcommand{\Bool}{\ensuremath{\{0,1\}}}

\newcommand{\Z}{\mathbb{Z}}

\newcommand{\critr}{criterion\xspace}
\newcommand{\critrs}{criteria\xspace}

\newcommand{\ring}{commutative ring with identity}
\newcommand{\ringg}[1]{commutative ring \ensuremath{#1} with identity}

\newcommand{\CSP}{\ensuremath{\mathsf{CSP}}\xspace}
\newcommand{\csp}{\CSP}
\newcommand{\CSPs}{\ensuremath{\mathsf{CSPs}}\xspace}

\newcommand{\cspx}[1]{\ensuremath{\csp(#1)}}
\newcommand{\cspg}{\cspx{\Gamma}}
\newcommand{\cspr}{\cspx{R}}

\newcommand{\nr}[2]{\ensuremath{\operatorname{NRD}_{#2}(#1)}}

\newcommand{\nrp}[3]{\ensuremath{\operatorname{NRD}^{#3}_{#2}(#1)}}

\newcommand{\NRD}{\ensuremath{\mathsf{NRD}}\xspace}
\newcommand{\nrd}{\NRD}
\newcommand{\nrdx}[1]{\ensuremath{\nr{#1}{n}}}
\newcommand{\nrdg}{\nrdx{\Gamma}}
\newcommand{\nrdr}{\nrdx{R}}

\newcommand{\NAE}{\ensuremath{\mathsf{NAE}}}
\newcommand{\ext}{\ensuremath{\operatorname{ext}}}

\newcommand{\wt}[1]{\ensuremath{\operatorname{W}(#1)}}
\newcommand{\wtr}{\wt{R}}

\newcommand{\pppol}{\ensuremath{\operatorname{p^2Pol}}} %partial pattern polymorphism

\newcommand{\OR}{\ensuremath{\operatorname{OR}}\xspace}

\newcommand{\Sat}{\mathrm{Sat}}
\newcommand{\captures}{\mathrm{captures}}
\newcommand{\tbalanced}{\mbox{$t$-balanced}}
\newcommand{\tbalancedness}{\mbox{$t$-balancedness}}

\newcommand{\SAT}{\ensuremath{\mathsf{SAT}}\xspace}



\allowdisplaybreaks

\ifnum\doubleblind=0
\title{Non-Redundancy of Low-Arity Symmetric Boolean CSPs}
\author{Amatya Sharma\thanks{University of Michigan, Ann Arbor, Michigan, USA. Email: \texttt{amatya@umich.edu}. Work done in part when the author was visiting the Toyota Technological Institute at Chicago as an intern in Fall 2025.}\and{Santhoshini Velusamy \thanks{University of Waterloo, Ontario, Canada. Supported by NSERC Discovery grant RGPIN-2026-04689 and ECR supplement DGECR-2026-00421. Supported in part by NSF award CCF 2348475 when the author was affiliated with Toyota Technological Institute at Chicago. Email: \texttt{santhoshini.velusamy@uwaterloo.ca}}}}
\fi
\ifnum\doubleblind=1
\author{Anonymous authors}
\fi

\date{}
\usepackage{mathpazo}
\usepackage{microtype}
\begin{document}
\maketitle
\begin{abstract}

Non-redundancy, introduced by Bessiere, Carbonnel, and Katsirelos (AAAI 2020), has emerged as a central structural parameter for Constraint Satisfaction Problems (\textsf{CSPs}): it underlies kernelization and exact sparsification, and recent work shows that it also governs approximate sparsification and exact streaming complexity. The non-redundancy \textsf{NRD} measures the largest size of a \textsf{CSP}\ instance  that admits no smaller subinstance with the same set of satisfying assignments.

In this paper we study the non-redundancy ($\mathsf{NRD}_n(R)$) for Boolean symmetric \textsf{CSPs} specified by a symmetric $r$-ary relation $R$, that is, a relation whose value depends only on Hamming weight. An instance of \textsf{CSP}$(R)$ consists of $n$ variables and $m$ constraints, each given by an $r$-tuple of variables, and a constraint is satisfied if and only if the resulting tuple belongs to $R$. This is a natural and broad class containing \textsf{Max-Cut}, \textsf{Max-$k$-SAT}, etc. Our main result is a near-complete classification of the growth of $\mathsf{NRD}_n(R)$ for symmetric Boolean predicates of arity at most $5$. Using computational experiments together with algebraic upper- and lower-bound criteria, we determine the asymptotic growth of $\mathsf{NRD}_n(R)$ for every symmetric Boolean predicate of arity at most $4$. At arity $5$, the framework resolves all predicates except two.
%, namely those corresponding to the weight sets $\{0,2,3\}$ and $\{1,2,4\}$,for which we obtain an $\Omega(n^2)$ lower bound and an $O(n^3)$ upper bound.

On the upper-bound side, we introduce a lifted, higher-degree version of the balancedness notion of Chen, Jansen, and Pieterse (Algorithmica 2020), and show that it yields general upper bounds on non-redundancy. More precisely, we define $t$-balancedness and prove that it is equivalent to the existence of degree at most $t$ multilinear polynomials ``capturing'' $R$, implying that $\mathsf{NRD}_n(R)=O(n^t)$ for every $t$-balanced relation $R$. On the lower-bound side, we use Carbonnel's  (CP 2022) framework, which gives a criterion for when a predicate admits a special reduction from the $k$-ary OR relation; since OR already has non-redundancy $\Omega(n^k)$, this yields corresponding lower bounds on $\mathsf{NRD}_n(R)$.
While these determine the non-redundancy of symmetric predicates up to arity $4$ and nearly all predicates of arity $5$, we discover two exceptional predicates of arity $5$ that are not characterized by this framework---we obtain an $\Omega(n^2)$ lower bound and an $O(n^3)$ upper bound. We further show that these remaining cases reduce to natural extremal set-system questions.

\paragraph{Independent Work. } Very recently, Brakensiek, Guruswami, and Putterman \cite{brakensiek2026classification} used a similar experimental viewpoint to study arity-$4$ Boolean predicates, and one of their remaining cases connects to the symmetric arity-$5$ predicate with W$(R_{023})=\{0,2,3\}$ that also appears here.

\end{abstract}

\newpage
\tableofcontents

\section{Introduction}

Constraint satisfaction problems ($\CSP$s) are a central framework in theoretical computer science. Formally, given a finite domain $D$ and a constraint language $\Gamma$ consisting of $r$-ary relations over
$D$, an instance of $\CSP(\Gamma)$ consists of $n$ variables $V := \{x_1,\ldots,x_n\}$ together with
constraints of the form $\{R,(\bar x)\}$, where $R\in \Gamma$ and $\bar x\in V^r$.
\footnote{
Allowing unary maps on variables (e.g. literals in the Boolean case), constants,
or repeated variables inside constraints does not change the asymptotic
non-redundancy; see \Cref{sec:positive-nrd}.
}
The goal is to decide whether there exists an assignment $\sigma :  V \to D$ such that $(\sigma(x_1),\dots,\sigma(x_r)) \in R$ for every constraint. This framework captures many familiar problems, including $\SAT$, graph coloring, and systems of local equations, and has been studied from the viewpoints of complexity theory 
%\cite{Schaefer78,Chen06SchaeferExposition}\amati{add more citations}
\cite{Schaefer78,jeavons1997closure,feder1998computational,Hastad01,Khot02,bulatov2005classifying,Chen06SchaeferExposition,zhuk2020proof}, 
approximation \cite{khanna1997complete,khanna2001approximability,lewin2002improved,bazgan2003polynomial,ACMM05,cohen2006complexity,KKMO07,raghavendra2008optimal,KS09,austrin2009approximation,guruswami2011lasserre,manurangsi2015approximating,fotakis2016sub,DBLP:conf/stoc/KhotTW14,CLRS16,GT17,DBLP:conf/focs/AlevJT19,jeronimo2020unique,bafna2021playing,jeronimo2021near,als25,alms25},
(approximate) sparsification 
\cite{kogan2015sketching,filtser2017sparsification,butti2020sparsification,KPS24,khanna2025efficient,khanna2025theory}%\cite{khanna2001approximability,kogan2015sketching,filtser2017sparsification,butti2020sparsification,khanna2025efficient,khanna2025theory}
% \cite{khanna1997complete,khanna2001approximability,Hastad01,Khot02,KKMO07,raghavendra2008optimal,austrin2009approximation,lewin2002improved},
% Approximation / Max-CSP background
,
kernelization and exact sparsification \cite{dell2014satisfiability,jansen2019optimal,chen2020best,lagerkvist2020sparsification, bessiere2020chain,carbonnel2022redundancy,brakensiek2025redundancy,brakensiek2025richness,brakensiek2025tight}, and, more recently, streaming and sketching \cite{KKS15,GVV17,KK19,CGV20,CGS+22-linear-space,Sin23-kand,KolParamonovSaxenaYu-ITCS23,Vu-2208.09160,CGSV24,STV25,FMW26,FMW26-tight,ABF26,STV26}.

\paragraph{Symmetric Boolean \CSPs.}
A relation $R \subseteq \{0,1\}^r$ is symmetric if membership in $R$ depends only on the Hamming weight of the input. Equivalently, there is a set $\wtr \subseteq \{0,1,\dots,r\}$ such that $R = \{x \in \{0,1\}^r : |x| \in \wt{R}\}$.
Thus a symmetric predicate can be specified simply by its set of accepted weights, and throughout the paper we will freely identify a symmetric predicate with its weight set $\wtr$.

% There are several standard conventions for constraints over $R$: one may
% allow literals and constants, allow only variables with repetitions, or
% restrict to simple positive constraints with distinct variables. For fixed
% arity, these choices are equivalent for our purposes: Section~6 shows that
% they have the same asymptotic non-redundancy up to a constant-factor blow-up
% in the number of variables. Thus the exponent of $\nrdr$ is independent
% of these modeling choices, and we freely move between the models when
% convenient.

Symmetric Boolean predicates already contain many natural and canonical problems. For example, the $r$-ary \OR predicate corresponds to $\wt{\OR_r}=\{1,\dots,r\}$ and gives the usual $r$-\SAT\ constraints; the $r$-ary \textsf{Not-All-Equal} predicate corresponds to $\wt{\NAE_r}=\{1,\dots,r-1\}$ and is closely related to hypergraph $2$-coloring and $\NAE$-$\SAT$; threshold predicates correspond to intervals of the form $\{t,t+1,\dots,r\}$ and capture ``at least $t$'' type constraints; and parity predicates correspond to the even or odd weights and encode systems of linear equations over $\mathbb{F}_2$. For these reasons, symmetric predicates form a broad yet structured class, and they have been studied explicitly in work on approximation, sparsification, sketching, and non-redundancy \cite{austrin2012quadratic,guruswami2017towards,BHP+22, brakensiek2021promise, brakensiek2023sdps,khanna2025efficient}.

Thus, we restrict our study to symmetric Boolean \CSPs in this paper. We ask how large can an instance be if every one of its constraints is non-redundant, that is, if no constraint can be removed without changing the satisfying set. This leads to the notion of non-redundancy.

\paragraph{Non-redundancy.} Given a \CSP\ instance, a constraint is \emph{redundant} if deleting it does not change the set of satisfying assignments. An instance is \emph{non-redundant} if every constraint is \emph{non-redundant} in this sense. Following Bessiere, Carbonnel, and Katsirelos \cite{bessiere2020chain}, who introduced this viewpoint in the context of learning \CSP\ instances with few membership queries, and Carbonnel's \cite{carbonnel2022redundancy} subsequent structural study, we define the non-redundancy dimension of a constraint language $\Gamma$ by
\[
\nrdg := \max\{ |I| : I \text{ is a non-redundant } \cspg\text{ instance on } n \text{ variables}\}.
\]
When $\Gamma=\{R\}$ consists of a single relation, we also write $\nrdr$. Thus $\nrdr$ asks a basic extremal question: how many $R$-constraints can be simultaneously non-redundant in an instance on $n$ variables? At first sight, this is purely a combinatorial quantity. However, it turns out to be closely connected to several algorithmic and structural questions about \CSPs.

\paragraph{Connections of \NRD to other problems.}

One such connection comes from \emph{exact sparsification}. The goal is to reduce any $n$-variable instance in polynomial time to an equivalent subinstance of size $g(n)$, where equivalence means having exactly the same satisfying assignments. Exact sparsification is closely related to kernelization, where one preprocesses an instance to an equivalent instance whose size is bounded in terms of a chosen parameter. In our setting the parameter is the number $n$ of variables, and exact sparsification asks for the stronger property that the output is an equivalent subinstance obtained by deleting constraints. %\snote{This sentence is a sudden jump and conveys no meaning unless we define kernelization and are more clear what parameter we are referring to, etc.\amati{I see. can remove it. Just that when we say sparsificationm it is usually restricted to ``subinstace'' whereas kernelization need not keep a ``subisntance''}}. 
On the positive side, several structured Boolean \CSPs admit exact sparsification by retaining only a small subset of constraints \cite{jansen2019optimal,chen2020best,lagerkvist2020sparsification}. Non-redundant instances are the canonical obstruction to such subinstance sparsification: any equivalent subinstance must retain every constraint of a non-redundant input.

% One connection comes from exact sparsification. The goal is to reduce an \(n\)-variable instance in polynomial time to a small equivalent subinstance, where equivalence means having exactly the same satisfying assignments. This is closely related to kernelization, where one preprocesses an instance to an equivalent instance whose size is bounded in terms of a chosen parameter; here the parameter is the number \(n\) of variables, and exact sparsification asks for the stronger property that the output is obtained by deleting constraints. Non-redundant instances are the canonical obstruction to such subinstance sparsification: any equivalent subinstance of a non-redundant input must retain every constraint.

Another, more approximation-theoretic notion is \emph{approximate sparsification}, where one seeks a weighted subinstance of small size, with weights on constraints, that approximately preserves the value of every assignment. This line of work originates in graph and hypergraph cut sparsification and has since developed into a broad theory for \CSPs \cite{kogan2015sketching,filtser2017sparsification,butti2020sparsification,khanna2025efficient,khanna2025theory}. Very recently, Brakensiek and Guruswami \cite{brakensiek2025redundancy} proved that, for every unweighted predicate $R$, the worst-case size of approximate sparsifiers is governed, up to polylogarithmic factors, by $\nrdr$; see also the recent follow-up on random \CSP\ sparsification \cite{brakensiek2025tight}. Thus the same combinatorial parameter that obstructs exact sparsification also controls approximate sparsification.

Recently, there has been a similar connection in the \emph{streaming} setting as well. In very recent work by Sharma and Velusamy \cite{sharma2026characterizingstreamingdecidabilitycsps}, it was shown that, up to logarithmic factors, the single-pass streaming complexity of deciding satisfiability and exact sparsification is characterized by non-redundancy. Taken together, these results show that $\NRD$ is not merely a technical parameter arising in one corner of the literature, but rather a structural invariant that governs exact sparsification, approximate sparsification, and streaming decidability and streaming exact sparsification.

These connections motivate the problem studied in this paper: 

\begin{center}
\noindent\emph{Determine the growth of $\nrdr$ for symmetric Boolean predicates $R$.}
\end{center}

Our goal is both structural and computational: to identify algebraic \critrs that certify polynomial upper and lower bounds on $\NRD_n(R)$, and to use these criteria to classify low-arity symmetric Boolean predicates.
\subsection{Our results}

Our main result is a near-complete classification of the growth of $\nrdr$ for low-arity symmetric Boolean predicates. For arities at most $4$, our upper- and lower-bound \critrs always match, giving the exact asymptotic growth. At arity $5$, the same framework resolves all predicates except two.

\begin{theorem}
For every non-trivial symmetric Boolean predicate $R$, the following hold.
\begin{enumerate}
    \item If $R$ has arity at most $4$, then the asymptotic growth of $\nrdr$ is determined exactly.
    \item If $R$ has arity $5$, then the asymptotic growth of $\nrdr$ is determined exactly for all predicates except, up to bit-flip symmetry,\footnote{Bit-flip symmetry identifies $R\subseteq\{0,1\}^r$ with $\overline R:=\{1-x:x\in R\}$. For symmetric predicates, this sends $W$ to $r-W:=\{r-w:w\in W\}$ and preserves $\NRD_n$. See \Cref{sec:prelims}.} the two predicates $R_{023}$ and $R_{124}$ with $\wt{R_{023}}=\{0,2,3\}$ and $\wt{R_{124}}=\{1,2,4\}$. For each $R\in\{R_{023},R_{124}\}$, we obtain
$\Omega(n^2)\le \nrdr\le O(n^3)$.
\end{enumerate}
\end{theorem}

The exact exponents for all predicates covered by the classification are reported in \Cref{tab:arity1to4-experiments,tab:arity5-experiment}. The classification is obtained by combining an algebraic upper-bound \critr with an algebraic lower-bound \critr. For a relation $R$, let $u(R)$ denote the exponent certified by our upper-bound \critr, and let $\ell(R)$ denote the exponent certified by the lower-bound \critr. Whenever $u(R)=\ell(R)$, the two \critrs give matching bounds $\nrdr=\Theta(n^{u(R)})$.

At arity $5$; as suggested by the \Cref{tab:arity5-experiment}; our framework leaves exactly two unresolved predicates:
$\wt{R_{023}}=\{0,2,3\}$ and $\wt{R_{124}}=\{1,2,4\}$. For both predicates, our lower-bound \critr yields only an $\Omega(n^2)$ lower bound, while our upper-bound \critr gives an $O(n^3)$ upper bound. Thus these are the first symmetric Boolean predicates that fall strictly between the current lower- and upper-bound frameworks. We further show in \Cref{sec:exceptional5} that understanding these predicates leads naturally to extremal set-system questions, which we view as a promising direction for future work.

The classification is obtained by combining two algebraic \critrs, one for upper bounds and one for lower bounds.

\paragraph{Upper bounds via lifted balancedness.}
% On the upper-bound side, our main conceptual contribution is a higher-degree version of the balancedness framework of \cite{chen2020best}. Their notion of balancedness captures the degree-$1$ polynomial phenomenon underlying linear sparsification. We extend this viewpoint by applying balancedness after the degree-$t$ monomial lift.
% \begin{theorem}[Informal upper-bound \critr]
% Let $R\subseteq\{0,1\}^r$ and let $t\ge 1$. Then $R$ is $t$-balanced if and only if every falsifying tuple $f\notin R$ is captured by a multilinear polynomial of degree at most $t$ over some prime-power ring. Consequently, every $n$-variable instance of $\CSP(R)$ has an equivalent subinstance with $O(n^t)$ constraints, and hence $\nrdr=O(n^t)$.
% \end{theorem}
On the upper-bound side, our main conceptual contribution is a lifted version of \emph{balancedness}. The notion of balancedness was introduced by Chen, Jansen, and Pieterse \cite{chen2020best} in their study of exact sparsification of Boolean $\CSP$s, where it captures the degree-$1$ polynomial phenomenon underlying linear sparsification. Our notion of $t$-balancedness extends this viewpoint to higher-degree multilinear monomial lifts. We show that $t$-balancedness exactly characterizes the existence of degree-$\le t$ polynomial that \emph{captures} the predicate, i.e., the polynomial that separates falsifying tuples from satisfying tuples, and therefore yields general $O(n^t)$ upper bounds on $\nrdr$.

\begin{restatable}[Upper-Bound \critr; informal \Cref{thm:t-balanced-main}]{theorem}{infupper}\label{thm:intro-tbalanced}
% \begin{theorem}
Let $R \subseteq \{0,1\}^r$ and let $t \ge 1$. Then $R$ is $t$-balanced if and only if every falsifying tuple $f \notin R$ is captured by a multilinear polynomial of degree at most $t$ over some \ring. Consequently, every $n$-variable instance of $\cspr$ has an equivalent subinstance with $O(n^t)$ constraints, and hence $\nrdr=O(n^t)$. 
% \end{theorem}
\end{restatable}

% The implication from low-degree capture to exact sparsification builds on the polynomial sparsification framework of Chen, Jansen, and Pieterse \cite{chen2020best}. Thus $t$-balancedness gives a general and effective upper-bound \critr for non-redundancy.

\paragraph{Lower bounds via universal cubes.}
% On the lower-bound side, we use Carbonnel's framework for proving non-redundancy lower bounds~\cite{Car22}. The relevant consequence is that if $R$ can define the $k$-ary $\OR$ relation through Carbonnel's fgpp-definitions, then $R$ inherits the $\Omega(n^k)$ non-redundancy lower bound of $\OR_k$. Equivalently, in the Boolean setting this is certified by failure of the universal $k$-cube partial polymorphism.
% \begin{theorem}[Informal lower-bound criterion]
% Let $R\subseteq\{0,1\}^r$ and let $2\le k\le r$. If the universal $k$-cube operation does not preserve $R$, equivalently if $\OR_k\in\langle R\rangle_{\mathrm{fg}}$, then $\nrdr=\Omega(n^k)$.
% \end{theorem}
On the lower-bound side, we use Carbonnel's \cite{carbonnel2022redundancy} algebraic framework for proving lower bounds on non-redundancy. At a high level, this framework gives a \critr for when a relation $R$ is expressive enough to encode the $k$-ary $\OR$ relation. Since $\OR_k$ has non-redundancy $\Omega(n^k)$, witnessed by placing one positive $\OR_k$ constraint on every $k$-subset of variables. Therefore, such an encoding immediately yields an $\Omega(n^k)$ lower bound on $\nrdr$. For details see \Cref{sec:lower-bound}.

\paragraph{Robustness to instance-model choices.}
Our classification concerns the exponent of $\nrdr$ and is unaffected by
standard choices in the constraint model. In \Cref{sec:positive-nrd}, we prove
that allowing unary maps on variables, constants, or repeated variables inside
constraints does not change the asymptotic non-redundancy, compared with the
model where each constraint is applied directly to a tuple of distinct variables.
The equivalence holds up to constant-factor changes in the number of variables.
The intuition is that, for Boolean predicates of fixed arity, literals,
constants, and repeated occurrences can be simulated by replacing each original
variable with a constant number of tagged copies. Hence our results are
independent of these modelling choices.
% \paragraph{Robustness to instance-model choices.}
% Our classification concerns the exponent of $\nrdr$ and is unaffected by
% standard choices in the constraint model. In \Cref{sec:positive-nrd}, we prove
% that allowing unary maps on variables, constants, or repeated variables inside
% constraints does not change the asymptotic non-redundancy, compared with the
% model where each constraint is applied directly to a tuple of distinct variables.
% The intuition is that, for Boolean predicates of fixed arity, literals,
% constants, and repeated occurrences can be simulated by replacing each original
% variable with a constant number of tagged copies. Hence our results are independent of these modelling choices.

% Very recently, Brakensiek, Guruswami, and Putterman used a similar experimental viewpoint to study low-arity Boolean predicates, and one of their remaining cases reduces to the symmetric arity-$5$ predicate with $\wt{R}=\{0,2,3\}$ that also appears here \cite{brakensiek2025richness}.\snote{It's okay this skip this here as it already appears in the abstract.}

\subsection{Organization}
In \Cref{sec:prelims}, we introduce basic notation, the \CSP\ model, symmetric Boolean predicates, and the different instance models used in the paper. In \Cref{sec:experiments}, we present the computational framework: the upper-bound \critr based on \tbalancedness{}, the lower-bound \critr based on the universal-cube test, the resulting classification up to arity $4$, and the two exceptional arity-$5$ predicates. In \Cref{sec:upper-bound}, we prove the upper-bound \critr by relating \tbalancedness{} to low-degree capturing polynomials and then applying the sparsification theorem of \cite{chen2020best}. In \Cref{sec:lower-bound}, we recall Carbonnel's\cite{carbonnel2022redundancy} lower-bound framework and proves the universal-cube lower-bound \critr used in the computations. In \Cref{sec:positive-nrd}, we prove the robustness of non-redundancy under standard instance-model choices, such as allowing unary maps, constants, and repeated variables in constraints. Finally, in \Cref{sec:conclusion}, we conclude with some open problems. 

\section{Preliminaries and notation}\label{sec:prelims}

% \snote{Condense everything to a single paragraph instead of starting a separate paragraph for each sentence.}
\subsection{Basic notation}

For a positive integer $n$, we write $[n] := \{1,2,\dots,n\}$. %We let $\Bool := \{0,1\}$. 
For a finite set $S$, we write $\Bool^S$ for the set of all Boolean vectors indexed by $S$, that is, all functions $x \colon S \to \Bool$. When $S = [n]$, we identify $\Bool^{[n]}$ with the usual Boolean cube $\Bool^n$.
For a vector $x \in \Bool^r$, we write $|x|$ for its Hamming weight, that is, the number of coordinates equal to $1$.
All asymptotic notation is with respect to the number $n$ of variables. The hidden constants may depend on fixed parameters such as the arity of the relation under consideration.

% \paragraph{Polynomials over Boolean variables.}

Throughout, all rings are commutative rings with identity. When evaluating
polynomials over a ring $E$ on Boolean inputs, we identify $0,1$ with
$0_E,1_E$. Integer coefficients and signs are interpreted through the
canonical map $\mathbb Z\to E$, $k\mapsto k\cdot 1_E$. By a prime-power
ring, we mean the residue ring $\mathbb Z/q\mathbb Z$, where $q$ is a power of some prime $p$.

We will use polynomials in variables $x_1,\ldots,x_r$ over such a ring $E$. Since these polynomials are only evaluated on Boolean inputs, every polynomial can be multilinearized using the identities $x_i^2=x_i$. Thus we may assume without loss of generality that all polynomials are multilinear, i.e., of the form $p(x)=\sum_{S\subseteq [r]} \alpha_S \prod_{i\in S}x_i$ with coefficients $\alpha_S\in E$. The degree of a multilinear polynomial is the maximum $|S|$ such that $\alpha_S\neq 0$.

\subsection{\textsf{CSPs}}

A constraint satisfaction problem is specified by a set of relations over a domain. Fix a finite domain $D$ and a set $\Gamma$ of $r$-ary relations over $D$. 
An instance of $\cspg$ consists of $n$ variables $V:=\{x_1, \ldots, x_n\}$ and a set $I$ of $m$ constraints. A constraint $C$ is defined by a pair $\{R,\bar x\}$, where $R\in\Gamma$ is an $r$-ary relation and $\bar x=(x_{i_1},\ldots,x_{i_r})\in V^r$ is an ordered tuple of variables. Repetitions inside $\bar x$ are allowed unless explicitly stated otherwise.
An assignment is a map $\sigma\colon V\to D$. Such an assignment satisfies the constraint $\{R,(x_{i_1},\ldots,x_{i_r})\}$ if $(\sigma(x_{i_1}),\ldots,\sigma(x_{i_r}))\in R$. We write $\Sat(I)$ for the set of satisfying assignments of an instance $I$.

Notice that above we assume that all relations in a fixed language $\Gamma$ have the same arity $r$. This is without loss of generality for fixed finite languages: if $R\subseteq D^k$ has arity $k<r$, we may replace it by the padded relation
$R' := R\times D^{r-k}\subseteq D^r$, where the extra coordinates are dummy coordinates. A constraint using $R$ can then be viewed as a constraint using $R'$ by adding arbitrary variables in the dummy positions.

% \paragraph{Boolean \CSPs.}\amati{i removed literals from here, since i discuss it in the later paragraph}
% In this paper we are interested in Boolean \CSPs, so from now on the domain will be $\Bool$ unless stated otherwise.
%A literal is either a variable $x$ or its negation $\neg x$. If $\sigma \colon V \to \Bool$ is an assignment, we extend it to literals in the standard way by setting $\sigma(\neg x) := 1-\sigma(x)$. Thus, if $R \subseteq \Bool^r$ is a Boolean relation and $(\ell_1,\dots,\ell_r)$ is an ordered $r$-tuple of literals, then the constraint $\{R,(\ell_1,\dots,\ell_r)\}$ is satisfied by $\sigma$ if $(\sigma(\ell_1),\dots,\sigma(\ell_r)) \in R$. We use the terms Boolean relation and Boolean \emph{predicate} interchangeably.%: a relation $R\subseteq\{0,1\}^r$ is identified with its indicator predicate $R:\{0,1\}^r\to\{0,1\}$.

\paragraph{Symmetric Boolean predicates and the \OR relation.}

From now on, when discussing \emph{Boolean relations}, the domain is $\{0,1\}$. We use the terms Boolean relation and \emph{Boolean predicate} interchangeably. Let $R\subseteq\{0,1\}^r$ be a Boolean relation. We say that $R$ is \emph{symmetric} if membership in $R$ depends only on Hamming weight. Equivalently, there is a set $\wtr \subseteq \{0,1,\dots,r\}$ such that $R = \{x \in \Bool^r : |x| \in \wt{R}\}$. Throughout the paper, we identify a symmetric relation with its accepted-weight set $\wtr$.

For $k \ge 2$, we write $\OR_k \subseteq \Bool^k$ for the $k$-ary \OR relation, that is, $\OR_k := \Bool^k \setminus \{0^k\}$. Equivalently, $(x_1,\dots,x_k)\in \OR_k$ if and only if at least one input bit is $1$. Equivalently, $\wt{\OR_k}=\{1,\ldots, k\}$.

Although one can define non-redundancy for arbitrary finite Boolean constraint
languages, it is enough, at the level of exponents, to understand singleton
languages $\Gamma=\{R\}$. Indeed, if $\Gamma=\{R_1,\ldots,R_s\}$ is fixed and
the singleton exponents are known, then the exponent for $\Gamma$ is the
maximum of the singleton exponents. The lower bound follows by using only one
relation, and the upper bound follows by sparsifying the constraints of each
relation separately. We therefore mostly work with a single symmetric relation
$R\subseteq \Bool^r$.

\paragraph{Equivalence and non-redundancy.}
Two instances $I$ and $I'$ are \emph{equivalent} if they have the same set of satisfying assignments, that is, if $\Sat(I)=\Sat(I')$.
Let $I$ be an instance and let $C$ be a constraint of $I$. We say that $C$ is \emph{redundant} if $I \setminus \{C\}$ is equivalent to $I$, that is, if deleting $C$ does not change the set of satisfying assignments. We say that $I$ is \emph{non-redundant} if none of its constraints is redundant.
The main parameter studied in this paper is the non-redundancy of a language. For a finite language $\Gamma$, define $\nrdg$ to be the maximum number of constraints in a non-redundant instance of $\cspg$ on $n$ variables.
When the language consists of a single relation $R$, we write $\Gamma_R := \{R\}$, $\cspr := \CSP(\Gamma_R)$, and $\nrdr := \nrdx{\Gamma_R}$.

We call a relation $R\subseteq \Bool^r$ \emph{non-trivial} if
$\emptyset\neq R\subsetneq \Bool^r$. The two trivial cases are easy to handle
separately: if $R=\Bool^r$, then every $R$-constraint is always satisfied and
hence redundant, so $\nrdr=0$; if $R=\emptyset$, then any instance containing at
least one $R$-constraint is unsatisfiable, and it is easy to see that
$\nrdr=1$. Thus, except where stated otherwise, we restrict attention to
non-trivial relations.

We also note that every non-trivial relation already has linear non-redundancy. Indeed, if $a \in R$ and $b \in \Bool^r \setminus R$, then by partitioning the variables into $\lfloor n/r \rfloor$ pairwise disjoint $r$-tuples and placing one copy of $R$ on each block, we obtain a non-redundant instance with $\lfloor n/r \rfloor$ constraints: to witness that a given block is essential, assign $b$ on that block and $a$ on every other block. Hence $\nrdr=\Omega(n)$ for every non-trivial $R$.

\paragraph{Bit-flip symmetry.}
For a Boolean relation $R\subseteq\{0,1\}^r$, let $\overline R:=\{1-x:x\in R\}$ be its image under global bit-flip. If $R$ is symmetric with accepted weights $W(R)$, then $\overline R$ has accepted weights $r-W(R):=\{r-w:w\in W(R)\}$. The relations $R$ and $\overline R$ have the same non-redundancy: replacing every $R$-constraint by the corresponding $\overline R$-constraint and complementing all Boolean assignments gives a bijection between satisfying assignments, preserving equivalence of subinstances and non-redundancy. Hence $\nrdr=\nrdx{\overline R}$.

\paragraph{Literals, positive, and simple positive models.}

We will use two closely related models of instances, and one useful restriction of the positive model.

\begin{enumerate}
    \item \emph{Literal model.} Fix a finite set $\mathcal U$ of unary maps $g:D\to D$ containing the identity map. Constraints may be applied to \emph{literals}, i.e., unary substitutions of variables: a constraint has the form
    $\{R,(g_1(x_{i_1}),\ldots,g_r(x_{i_r}))\}$,
    where $R\subseteq D^r$ and $g_1,\ldots,g_r\in\mathcal U$. Constant maps may be included in $\mathcal U$, so this model also allows constants. We call this the literal model because, in the Boolean case, taking $\mathcal U=\{\operatorname{id},\neg,0,1\}$ gives exactly the usual variables, negated variables, and constants.

    \item \emph{Positive model.} Constraints are applied only to variables: a constraint has the form
    $\{R,(x_{i_1},\ldots,x_{i_r})\}$,
    where the tuple $(x_{i_1},\ldots,x_{i_r})$ is ordered and repetitions are allowed. It is called \emph{simple} if no constraint repeats a variable. Thus, in a simple positive instance, every constraint has distinct variables. %The tuple is still ordered; however, when $R$ is symmetric, the order is immaterial, and we may identify a simple positive $r$-ary constraint with an $r$-element subset of variables.
\end{enumerate}

By \Cref{thm:positive-nrd}, for fixed finite domains and fixed arity, the literal, positive, and simple positive models have the same asymptotic non-redundancy up to a constant-factor blow-up in the number of variables. We will therefore move between them when convenient, explicitly indicating the model when it matters. Unless stated otherwise, $\cspr$ and $\nrdr$ refer to the literal model.
% For Boolean relations, we will use three closely related models of instances. In the \emph{literal model}, constraints may be applied to literals and constants: a constraint has the form $\{R,(\ell_1,\ldots,\ell_r)\}$, where each $\ell_i$ is a variable, a negated variable, or one of the constants $0,1$. In the \emph{positive model}, constraints are applied only to variables, so constraints have the form $\{R,(x_{i_1},\ldots,x_{i_r})\}$, where repetitions among the variables are allowed. In the \emph{simple positive model}, constraints are positive and the variables $x_{i_1},\ldots,x_{i_r}$ are required to be distinct. The tuple is still ordered; however, when $R$ is symmetric, the order is immaterial, and we may identify a simple positive $r$-ary constraint with an $r$-element subset of variables.

% By \Cref{thm:positive-nrd}, these three models have the same asymptotic non-redundancy up to a constant-factor blow-up in the number of variables. We will therefore move between these models when convenient, explicitly indicating the model when it matters. Unless stated otherwise, $\cspr$ and $\nrdr$ refer to the literal model.

\section{Computing non-redundancy of symmetric predicates of small arity}
    \subsection{Experimental methodology}\label{sec:experiments}

Our computational approach studies the lower-bound and upper-bound mechanisms separately.
For each symmetric Boolean relation $R$, we compute one algebraic witness for a lower bound on $\nrdr$ and one algebraic witness for an upper bound on $\nrdr$. When the two exponents match, this yields the exact growth rate of the non-redundancy. The full proofs appear later in \Cref{sec:lower-bound,sec:upper-bound}; here we state the statements and the intuition behind the two tests.

\subsubsection{Upper-bound \critr}

The upper-bound \critr is based on polynomial capturing. Let $R\subseteq \{0,1\}^r$ and let $f\in \{0,1\}^r\setminus R$ be a falsifying tuple. A polynomial $p$ over a \ringg{E} \emph{captures} $f$ with respect to $R$ if $p(y)=0$ for every $y\in R$, but $p(f)\neq 0$. Thus, a capturing polynomial separates one particular falsifying tuple from all satisfying tuples of $R$.

Capturing polynomials are useful for non-redundancy because of the sparsification theorem of Chen, Jansen, and Pieterse~\cite{chen2020best}: if every falsifying tuple $f\notin R$ can be captured by a polynomial of degree at most $t$, then every instance of $\cspr$ on $n$ variables has an equivalent subinstance with $O(n^t)$ constraints. Consequently, any non-redundant instance itself has size at most $O(n^t)$. Therefore, to prove $\nrdr=O(n^t)$, it is enough to certify that every falsifying tuple of $R$ has a degree-$\le t$ capturing polynomial.

We now introduce the degree-$t$ lift, whose purpose is to turn degree-$t$ multilinear polynomials into linear polynomials. Fix $t\ge 1$, and let $\mathcal S_{r,t}:=\{S\subseteq [r]:1\le |S|\le t\}$. The degree-$t$ lift maps each tuple $x\in\{0,1\}^r$ to the vector $\nu_t(x)\in\{0,1\}^{\mathcal S_{r,t}}$ whose $S$-coordinate is $\nu_t(x)_S=\prod_{i\in S}x_i$. Thus a degree-$t$ multilinear polynomial in $x_1,\ldots,x_r$ is exactly a linear polynomial in the lifted coordinates, together with a constant term.

\begin{example}\label{ex:nae-upper-1}
Consider $R=\NAE_3$, so $\wt{R}=\{1,2\}$. The falsifying tuples are $000$ and $111$. The polynomial
$p(x)=(x_1+x_2+x_3-1)(x_1+x_2+x_3-2)$
vanishes on all tuples of Hamming weight $1$ or $2$, and is nonzero on tuples of Hamming weight $0$ or $3$. Thus this degree-$2$ polynomial captures both falsifying tuples.

For $t=2$, the lift is $\nu_2(x)=(x_1,x_2,x_3,x_1x_2,x_1x_3,x_2x_3)$. The lifted relation $R^{(2)}$ consists of the six vectors obtained from the accepting tuples:
\[
\begin{array}{c|c}
x & \nu_2(x) \\ \hline
100 & (1,0,0,0,0,0) \\
010 & (0,1,0,0,0,0) \\
001 & (0,0,1,0,0,0) \\
110 & (1,1,0,1,0,0) \\
101 & (1,0,1,0,1,0) \\
011 & (0,1,1,0,0,1).
\end{array}
\]
The falsifying tuples lift to $\nu_2(000)=(0,0,0,0,0,0)$ and $\nu_2(111)=(1,1,1,1,1,1)$. After multilinearization, $p$ becomes
$p(x)=2-2(x_1+x_2+x_3)+2(x_1x_2+x_1x_3+x_2x_3)$,
which is a linear expression in the lifted coordinates. This illustrates why degree-$t$ capturing polynomials can be studied through linear conditions after the degree-$t$ lift.
\end{example}

Instead of explicitly searching for capturing polynomials in the lifted space, we use an equivalent closure condition called \emph{\tbalancedness}. Let $X_t:=\nu_t(\{0,1\}^r)$ be the lifted Boolean cube and let $R^{(t)}:=\nu_t(R)$ be the lifted relation. We say that $R$ is \tbalanced{} if $R^{(t)}$ is closed under odd alternating sums inside $X_t$: for every odd integer $m$, every $a^{(1)},\ldots,a^{(m)}\in R^{(t)}$, and every $u:=a^{(1)}-a^{(2)}+\cdots-a^{(m)}\in X_t$, we also have $u\in R^{(t)}$.
% The lift turns degree-$t$ multilinear polynomials in the original variables into linear polynomials in the lifted coordinates. Thus, after lifting, the question becomes whether the accepted lifted points can be separated from the rejected lifted points by linear polynomials. The following closure condition is exactly the condition that makes this linear separation possible.
% In general, the degree-$t$ lift turns degree-$t$ multilinear polynomials in the original variables into linear polynomials in the lifted coordinates. Thus the question becomes whether rejected lifted points can be separated from the accepted lifted points by linear polynomials. The following closure condition captures when such separation is possible.
% Let $X_t:=\nu_t(\Bool^r)$ be the lifted Boolean cube and $R^{(t)}:=\nu_t(R)$ be the lifted relation. We say that $R$ is \emph{\tbalanced} if the lifted set $R^{(t)}$ is closed under odd alternating sums inside $X_t$: for every odd integer $m$, $a^{(1)},\dots,a^{(m)}\in R^{(t)}$, and 
% $u:=a^{(1)}-a^{(2)}+\cdots-a^{(m-1)}+a^{(m)}$ belongs to $X_t$, we also have $u\in R^{(t)}$.
The significance of this definition is the following theorem, proved in Section~\ref{sec:upper-bound}.

% \begin{theorem}[Upper-bound \critr, informal \Cref{thm:t-balanced-main}]
% %Let $R\subseteq \Bool^r$ and let $t\ge 1$. Then, 
% For a relation $R\subseteq\{0,1\}^r$ and an integer $t\ge 1$, $R$ is $t$-balanced iff every $f\notin R$ has a degree-$\le t$ capturing polynomial over some \ring. Consequently, every instance of $\cspr$ on $n$ variables has an equivalent subinstance with $O(n^t)$ constraints, and hence $\nrdr=O(n^t)$.
% \end{theorem}

\infupper*

\paragraph{Implementation.}
We now describe how the upper-bound \critr is implemented. Fix a symmetric relation $R\subseteq \Bool^r$ and an integer $t\ge 1$. We first form the degree-$t$ lift $\nu_t(x)$ for every tuple $x\in\Bool^r$. This gives the lifted Boolean cube $X_t=\nu_t(\Bool^r)$ and the lifted accepting set $R^{(t)}=\nu_t(R)$. The goal is to test whether $R^{(t)}$ is balanced relative to $X_t$.

Directly checking the definition of $t$-balancedness would require searching over all odd alternating sums of tuples from $R^{(t)}$. Instead, we use the equivalent integer-span formulation proved in \Cref{lem:bal-latt}. For a lifted vector $v$, write $\ext(v):=(1,v)$. Then $R$ is $t$-balanced if and only if, for every $b\in X_t\setminus R^{(t)}$, the vector $\ext(b)$ does not lie in the integer span of the vectors $\{\ext(a):a\in R^{(t)}\}$.

Thus the computation reduces to a finite set of integer-span tests. We form the integer matrix $M_R^{(t)}$ whose columns are the vectors $\ext(a)$ for $a\in R^{(t)}$. For each $b\in X_t\setminus R^{(t)}$, we test whether the integer linear system $M_R^{(t)}z=\ext(b)$ has a solution $z$ over $\mathbb Z$. If such a solution exists for some $b$, then $R$ is not $t$-balanced. If no such $b$ lies in the integer span, then $R$ is $t$-balanced.

We perform these integer-span tests using Smith normal form. Namely, we compute the Smith normal form of $M_R^{(t)}$ and use it to decide membership in the integer span of its columns. This avoids explicitly enumerating alternating sums and instead uses the integer-span characterization of balancedness. For each symmetric relation $R$, we test $t=1,2,\dots,r$ and record the smallest value $u(R)$ for which the test succeeds. By the upper-bound \critr above, this certifies $\nrdr=O(n^{u(R)})$.

\begin{example}[Continuing \Cref{ex:nae-upper-1}]
The previous example showed the polynomial side of the lift. We now use the same relation to illustrate the integer-span test actually used in the implementation. 

The two rejected lifted points are $\nu_2(000)=(0,0,0,0,0,0)$ and $\nu_2(111)=(1,1,1,1,1,1)$. To test $2$-balancedness, we check whether either rejected point lies in the integer span of the accepted lifted points after adding a leading $1$ coordinate. For $b=\nu_2(000)$, membership would mean that there are integers $z_{100},z_{010},z_{001},z_{110},z_{101},z_{011}$ such that
\[
\sum_{a\in R^{(2)}} z_a \ext(a)=\ext(b)=(1,0,0,0,0,0,0).
\]
Looking at the three degree-$2$ coordinates forces $z_{110}=z_{101}=z_{011}=0$. Then the three degree-$1$ coordinates force $z_{100}=z_{010}=z_{001}=0$, contradicting the leading-coordinate equation $\sum_a z_a=1$. Hence $\ext(\nu_2(000))$ is not in the integer span.

Similarly, for $b=\nu_2(111)$, the three degree-$2$ coordinates force $z_{110}=z_{101}=z_{011}=1$. The degree-$1$ coordinates then force $z_{100}=z_{010}=z_{001}=-1$. But the leading-coordinate sum is then $-1-1-1+1+1+1=0$, not $1$. Hence $\ext(\nu_2(111))$ is also not in the integer span. Therefore neither rejected lifted point lies in the integer span, so the integer-span test certifies that $\NAE_3$ is $2$-balanced.
\end{example}

% For the implementation, we use an equivalent integer-span formulation proved in Section~\ref{sec:upper}. For $v\in\Bool^{\mathcal S_{r,t}}$, write $\ext(v):=(1,v)$. Then $R$ is $t$-balanced if and only if for every $b\in X_t\setminus R^{(t)}$, the vector $\ext(b)$ does not lie in the integer span of $\{\ext(a):a\in R^{(t)}\}$. We test this lattice-membership condition using Smith normal form. For each relation $R$, we define $u(R)$ to be the smallest $t$ for which this upper-bound \critr succeeds.

\subsubsection{Lower bound \critr}
% The lower-bound \critr is based on reducing from the $k$-ary \OR predicate. The predicate $\OR_k$ has non-redundancy $\Theta(n^k)$: for example, the instance containing one positive $\OR_k$ constraint on each $k$-subset of variables is non-redundant. Thus, if each $\OR_k$-constraint can be replaced by a constant-size gadget over $R$ in a way that preserves non-redundancy witnesses, then $\nrdr=\Omega(n^k)$.

The lower-bound \critr is based on defining the $k$-ary \OR predicate from $R$. The predicate $\OR_k$ has non-redundancy $\Theta(n^k)$. For example, the instance containing one positive $\OR_k$ constraint on each $k$-subset of variables is non-redundant. Thus, if $\OR_k$ is ``definable'' from $R$ by a constant-size gadget, then non-redundant $\OR_k$-instances can be translated into non-redundant $R$-instances, giving $\nrdr=\Omega(n^k)$.

\begin{example}
    The predicate $\mathsf{NAE}_3$ defines $\OR_2$ after pinning one coordinate to $0$: the constraint $\{\mathsf{NAE}_3,(x,y,0)\}$ is false exactly on $(x,y)=(0,0)$ and true otherwise. Thus it behaves as $\{\OR_2,(x,y)\}$. Since $\OR_2$ has non-redundant instances of size $\Omega(n^2)$, this gives an $\Omega(n^2)$ lower bound.
\end{example}

The formal version of this gadget notion is Carbonnel's~\cite{carbonnel2022redundancy} \emph{fgpp-definability}. Informally, a relation $S$ is fgpp-definable from $R$ if every $S$-constraint can be expressed using a constant number of $R$-constraints, possibly with auxiliary variables and simple unary functional guards; see \Cref{def:fgppdef}. In particular, each occurrence of $S$ can be replaced by a constant-size gadget over $R$ in a way that preserves the non-redundancy lower bound. We denote this as $S\in\langle \{R\}\rangle_{\mathrm{fg}}$ or in short $S\in\langle R\rangle_{\mathrm{fg}}$.% as shorthand for $S\in\langle \{R\}\rangle_{\mathrm{fg}}$.

% \amati{i tried to informally define fgpp here, and skip the rest of the things, is it better? I could define all these in the prelims, but it might be overloading to ahead of time}The formal notion of definability we use is Carbonnel's \cite{carbonnel2022redundancy} fgpp-definability. Informally, a relation $S$ is fgpp-definable from $R$ if an $S$-constraint can be expressed using a constant number of $R$-constraints, possibly with auxiliary variables and simple unary functional guards; see formal \Cref{def:fgppdef}. Equivalently, each occurrence of $S$ can be replaced by a constant-size gadget over $R$ in a way that preserves the non-redundancy lower bound. We denote this as $S\in\langle \{R\}\rangle_{\mathrm{fg}}$ or in short $S\in\langle R\rangle_{\mathrm{fg}}$.% as shorthand for $S\in\langle \{R\}\rangle_{\mathrm{fg}}$.

\begin{theorem}[Lower-bound \critr; informal \Cref{thm:carbonnel-lower-bound}]
Let $R\subseteq\Bool^r$ and let $2\le k\le r$. If $\OR_k\in\langle R\rangle_{\mathrm{fg}}$, then $\nrdr=\Omega(n^k)$.
\end{theorem}

\paragraph{Implementation.}
% We do not search directly for an fgpp-definition of $\OR_k$ from $R$. Instead, we use Carbonnel's \cite{carbonnel2022redundancy} equivalent preservation \critr: $\OR_k\in\langle R\rangle_{\mathrm{fg}}$ if and only if a particular partial operation, called the universal $k$-cube operation, fails to preserve $R$. We denote this operation by $f_k$.

% Here preservation means the usual coordinatewise condition. Arrange tuples from $R$ as the rows of a matrix. If every column lies in the domain of $f_k$, then applying $f_k$ to each column gives an output tuple. The operation $f_k$ preserves $R$ if this output tuple always lies in $R$. Therefore a failure of preservation consists of input rows all lying in $R$ whose coordinatewise image under $f_k$ lies outside $R$.

We do not search directly for an fgpp-definition of $\OR_k$ from $R$. Instead, we use Carbonnel's \cite{carbonnel2022redundancy} equivalent preservation \critr: $\OR_k\in\langle R\rangle_{\mathrm{fg}}$ if and only if a certain partial operation, called the universal $k$-cube operation, does not preserve $R$. We denote this partial operation by $f_k$. Here ``preserve'' has the usual polymorphism meaning: if we arrange tuples from $R$ as rows of a matrix and apply $f_k$ coordinatewise, then the resulting tuple should again lie in $R$, whenever $f_k$ is defined on every column. Thus, to prove a lower bound, the code searches for a failure of preservation: input rows all lying in $R$ whose coordinatewise image under $f_k$ lies outside $R$.

Fix $k\ge 2$, and let $m=2^k-1$. Index the $m$ input rows by the nonzero strings $v\in\Bool^k\setminus\{0^k\}$. For each coordinate $i\in[k]$, define the column $c_i\in\Bool^m$ by $(c_i)_v=v_i$, where $v_i$ is the $i$th bit of $v$. The universal $k$-cube partial operation, denoted $f_k$, is the partial operation of arity $m$ whose domain is
$\{0^m,1^m\}\cup\{c_i,1-c_i:i\in[k]\}$,
and whose values are given by
$f_k(0^m)=0$, $f_k(1^m)=1$, $f_k(c_i)=0$, and $f_k(1-c_i)=1$ for every $i\in[k]$. See \Cref{ex:kcube}.

Concretely, to test whether $f_k$ preserves an $r$-ary relation $R$, choose $m$ tuples from $R$, one for each nonzero string $v\in\Bool^k\setminus\{0^k\}$, and place them as rows of an $m\times r$ matrix. If every column lies in the domain of $f_k$, then applying $f_k$ coordinatewise gives an output tuple in $\Bool^r$. A failure of preservation is exactly the case where all input rows lie in $R$, but this output tuple lies outside $R$.

\begin{example}\label{ex:kcube}
    For example, when $k=2$, the input rows are indexed by $11,10,01$. The operation $f_2$ has arity $3$ and is defined on the following column types:
\[
\begin{array}{c|c}
\text{column} & f_2(\text{column}) \\ \hline
000 & 0\\
111 & 1\\
110 & 0\\
001 & 1\\
101 & 0\\
010 & 1.
\end{array}
\]
Here $110$ and $101$ are the two coordinate columns of the nonzero $2$-cube, and $001$ and $010$ are their complements.
% \end{example}

% \begin{example}
As a concrete failure, take $R=\NAE_3$, so $\wt{R}=\{1,2\}$. Consider the following three input rows, indexed by $11,10,01$:
\[
\begin{array}{c|ccc|c}
 & \text{coord. }1 & \text{coord. }2 & \text{coord. }3 & \text{weight}\\ \hline
11 & 1 & 1 & 0 & 2\\
10 & 1 & 0 & 0 & 1\\
01 & 0 & 1 & 0 & 1.
\end{array}
\]
All three rows belong to $\NAE_3$. Now look at the columns: they are $110$, $101$, and $000$. By the table above, $f_2$ maps these columns to $0$, $0$, and $0$, respectively. Hence the coordinatewise output is $000$, which does not belong to $\NAE_3$. Therefore $f_2$ does not preserve $\NAE_3$, and thus gives the lower bound $\nr{\NAE_3}{n}=\Omega(n^2)$.
\end{example}

For symmetric relations, the implementation can simplify this preservation test. Since membership in $R$ depends only on Hamming weight, the actual positions of the column types do not matter; only their multiplicities matter. Thus, for fixed $R\subseteq\{0,1\}^r$ and $k$, the code enumerates all multiplicity vectors over the domain column types of $f_k$, with total multiplicity $r$. Each multiplicity vector determines the Hamming weights of all input rows and of the output row. The lower-bound \critr succeeds if every input-row weight lies in $\wt{R}$, while the output-row weight lies outside $\wt{R}$.

\begin{example}
Consider again the case $k=2$. The column types in the domain of $f_2$ are
$000,111,110,001,101,010$, with outputs $0,1,0,1,0,1$, respectively. Suppose $R=\NAE_3$, so $\wt{R}=\{1,2\}$. Since $R$ has arity $3$, a multiplicity vector assigns three coordinates among these six column types.

Take the multiplicity vector that uses $110$ once, $101$ once, and $000$ once, and uses all other column types zero times. The three input rows then have weights
$2,1,1$: the row indexed by $11$ sees the entries $1,1,0$, the row indexed by $10$ sees $1,0,0$, and the row indexed by $01$ sees $0,1,0$. All these weights lie in $\wt{R}=\{1,2\}$. The output row has entries $0,0,0$, hence output weight $0$, which does not lie in $\wt{R}$. Therefore this multiplicity vector witnesses failure of preservation, and the lower-bound \critr succeeds for $R=\NAE_3$ and $k=2$.
\end{example}

For each symmetric relation $R$, we run this test for $k=2,\dots,r$ and define $\ell(R)$ to be the largest $k$ for which the test succeeds, with the convention $\ell(R)=1$ if no such $k$ is found. By the lower-bound \critr above, this certifies $\nrdr=\Omega(n^{\ell(R)})$.

\paragraph{Enumerating all symmetric predicates.}
The code used for the computational experiments is publicly available at
\url{https://github.com/aaysharma/nrd_csps.git}.
Given an arity $r$, the code enumerates all symmetric Boolean $r$-ary relations, represented by their accepted-weight sets $\wt{R}\subseteq \{0,1,\dots,r\}$. For each relation, it computes the upper-bound exponent $u(R)$ certified by $t$-balancedness and the lower-bound exponent $\ell(R)$ certified by the universal $k$-cube test.

We enumerate relations only up to the bit-flip symmetry defined in
\Cref{sec:prelims}. For a symmetric arity-$r$ relation with accepted weights
$\wt{R}$, this symmetry sends $\wt{R}$ to
$r-\wt{R}:=\{r-w:w\in\wt{R}\}$ and preserves non-redundancy. Hence
$\nrdr=\nr{r-\wt{R}}{n}$, so it suffices to list one representative from each
pair $\wt{R}$ and $r-\wt{R}$. Modulo bit-flip symmetry, the number of relations is
$2^r+2^{\lceil (r-1)/2\rceil}$.

\subsection{A complete classification up to arity \texorpdfstring{$4$}{4}}
In all cases covered by our framework, the lower-bound exponent certified by the universal $k$-cube test matches the upper-bound exponent obtained from \tbalancedness, see \Cref{tab:arity1-experiment,tab:arity2-experiment,tab:arity3-experiment,tab:arity4-experiment}. However, for the arity-$5$ predicates we get a mismatch between the upper bound $O(n^3)$ and lower bound $\Omega(n^2)$, see \Cref{tab:arity5-experiment}. The specific \CSPs that have a gap and our attempt to prove tight bounds for those are discussed in \Cref{sec:exceptional5}.

\begin{table}[ht]
\centering
\scriptsize
\setlength{\tabcolsep}{3pt}
\renewcommand{\arraystretch}{0.92}

\begin{subtable}[t]{0.48\textwidth}
\centering
\begin{tabular}{r l r r r c}
\toprule
idx & $W$ & $|R|$ & $u(R)$ & $\ell(R)$ & mismatch \\
\midrule
1 & $\{0\}$ & 1 & 1 & 1 & no \\
\bottomrule
\end{tabular}
\caption{arity $r=1$.}
\label{tab:arity1-experiment}
\end{subtable}
\hfill
\begin{subtable}[t]{0.48\textwidth}
\centering
\begin{tabular}{r l r r r c}
\toprule
idx & $W$ & $|R|$ & $u(R)$ & $\ell(R)$ & mismatch \\
\midrule
1 & $\{0\}$ & 1 & 1 & 1 & no \\
2 & $\{1\}$ & 2 & 1 & 1 & no \\
3 & $\Bool$ & 3 & 2 & 2 & no \\
4 & $\{0,2\}$ & 2 & 1 & 1 & no \\
\bottomrule
\end{tabular}
\caption{arity $r=2$.}
\label{tab:arity2-experiment}
\end{subtable}

\vspace{0.75em}

\begin{subtable}[t]{0.48\textwidth}
\centering
\begin{tabular}{r l r r r c}
\toprule
idx & $W$ & $|R|$ & $u(R)$ & $\ell(R)$ & mismatch \\
\midrule
1 & $\{0\}$ & 1 & 1 & 1 & no \\
2 & $\{1\}$ & 3 & 1 & 1 & no \\
3 & $\Bool$ & 4 & 2 & 2 & no \\
4 & $\{0,2\}$ & 4 & 1 & 1 & no \\
5 & $\{1,2\}$ & 6 & 2 & 2 & no \\
6 & $\{0,1,2\}$ & 7 & 3 & 3 & no \\
7 & $\{0,3\}$ & 2 & 1 & 1 & no \\
8 & $\{0,1,3\}$ & 5 & 2 & 2 & no \\
\bottomrule
\end{tabular}
\caption{arity $r=3$.}
\label{tab:arity3-experiment}
\end{subtable}
\hfill
\begin{subtable}[t]{0.48\textwidth}
\centering
\begin{tabular}{r l r r r c}
\toprule
idx & $W$ & $|R|$ & $u(R)$ & $\ell(R)$ & mismatch \\
\midrule
1 & $\{0\}$ & 1 & 1 & 1 & no \\
2 & $\{1\}$ & 4 & 1 & 1 & no \\
3 & $\Bool$ & 5 & 2 & 2 & no \\
4 & $\{2\}$ & 6 & 1 & 1 & no \\
5 & $\{0,2\}$ & 7 & 2 & 2 & no \\
6 & $\{1,2\}$ & 10 & 2 & 2 & no \\
7 & $\{0,1,2\}$ & 11 & 3 & 3 & no \\
8 & $\{0,3\}$ & 5 & 1 & 1 & no \\
9 & $\{1,3\}$ & 8 & 1 & 1 & no \\
10 & $\{0,1,3\}$ & 9 & 2 & 2 & no \\
11 & $\{0,2,3\}$ & 11 & 2 & 2 & no \\
12 & $\{1,2,3\}$ & 14 & 3 & 3 & no \\
13 & $\{0,1,2,3\}$ & 15 & 4 & 4 & no \\
14 & $\{0,4\}$ & 2 & 1 & 1 & no \\
15 & $\{0,1,4\}$ & 6 & 2 & 2 & no \\
16 & $\{0,2,4\}$ & 8 & 1 & 1 & no \\
17 & $\{0,1,2,4\}$ & 12 & 3 & 3 & no \\
18 & $\{0,1,3,4\}$ & 10 & 2 & 2 & no \\
\bottomrule
\end{tabular}
\caption{arity $r=4$.}
\label{tab:arity4-experiment}
\end{subtable}

\caption{Outputs of the symmetric-relation experiments for arities $r=1,2,3,4$, excluding the trivial relations $\emptyset,\Bool^r$. The upper bounds $u(R)$ and lower bounds $\ell(R)$ match.}
\label{tab:arity1to4-experiments}
\end{table}

\begin{table}[ht]
\centering
\scriptsize

\begin{minipage}[t]{0.49\textwidth}
\centering
\begin{tabular}{r l r r r c}
\toprule
idx & $W$ & $|R|$ & $u(R)$ & $\ell(R)$ & mismatch \\
\midrule
% 0  & $\emptyset$        & 0  & 1 & 1 & no \\
1  & $\{0\}$            & 1  & 1 & 1 & no \\
2  & $\{1\}$            & 5  & 1 & 1 & no \\
3  & $\Bool$          & 6  & 2 & 2 & no \\
4  & $\{2\}$            & 10 & 1 & 1 & no \\
5  & $\{0,2\}$          & 11 & 2 & 2 & no \\
6  & $\{1,2\}$          & 15 & 2 & 2 & no \\
7  & $\{0,1,2\}$        & 16 & 3 & 3 & no \\
8  & $\{0,3\}$          & 11 & 1 & 1 & no \\
9  & $\{1,3\}$          & 15 & 2 & 2 & no \\
10 & $\{0,1,3\}$        & 16 & 2 & 2 & no \\
11 & $\{2,3\}$          & 20 & 2 & 2 & no \\
\textcolor{red}{12} & $\textcolor{red}{\{0,2,3\}}$        & \textcolor{red}{21} & \textcolor{red}{3} & \textcolor{red}{2} & \textcolor{red}{yes} \\
13 & $\{1,2,3\}$        & 25 & 3 & 3 & no \\
14 & $\{0,1,2,3\}$      & 26 & 4 & 4 & no \\
15 & $\{0,4\}$          & 6  & 1 & 1 & no \\
16 & $\{1,4\}$          & 10 & 1 & 1 & no \\
17 & $\{0,1,4\}$        & 11 & 2 & 2 & no \\
\bottomrule
\end{tabular}
\end{minipage}
\hfill
\begin{minipage}[t]{0.49\textwidth}
\centering
\begin{tabular}{r l r r r c}
\toprule
idx & $W$ & $|R|$ & $u(R)$ & $\ell(R)$ & mismatch \\
\midrule
18 & $\{0,2,4\}$        & 16 & 1 & 1 & no \\
\textcolor{red}{19} & $\textcolor{red}{\{1,2,4\}}$        & \textcolor{red}{20} & \textcolor{red}{3} & \textcolor{red}{2} & \textcolor{red}{yes} \\
20 & $\{0,1,2,4\}$      & 21 & 3 & 3 & no \\
21 & $\{0,3,4\}$        & 16 & 2 & 2 & no \\
22 & $\{0,1,3,4\}$      & 21 & 2 & 2 & no \\
23 & $\{0,2,3,4\}$      & 26 & 3 & 3 & no \\
24 & $\{1,2,3,4\}$      & 30 & 4 & 4 & no \\
25 & $\{0,1,2,3,4\}$    & 31 & 5 & 5 & no \\
26 & $\{0,5\}$          & 2  & 1 & 1 & no \\
27 & $\{0,1,5\}$        & 7  & 2 & 2 & no \\
28 & $\{0,2,5\}$        & 12 & 2 & 2 & no \\
29 & $\{0,1,2,5\}$      & 17 & 3 & 3 & no \\
30 & $\{0,1,3,5\}$      & 17 & 2 & 2 & no \\
31 & $\{0,2,3,5\}$      & 22 & 2 & 2 & no \\
32 & $\{0,1,2,3,5\}$    & 27 & 4 & 4 & no \\
33 & $\{0,1,4,5\}$      & 12 & 2 & 2 & no \\
34 & $\{0,1,2,4,5\}$    & 22 & 3 & 3 & no \\
% 35 & $\{0,1,2,3,4,5\}$  & 32 & 1 & 1 & no \\
\bottomrule
\end{tabular}
\end{minipage}

\caption{Output of the symmetric-relation experiment for arity $r=5$ excluding the trivial relations $\emptyset, \Bool^5$. The only mismatches are $\wt{R_{023}}=\{0,2,3\}$ and $\wt{R_{124}}=\{1,2,4\}$, as shown in red.}
\label{tab:arity5-experiment}
\end{table}

    \subsection{The exceptional arity-\texorpdfstring{$5$}{5} predicates}\label{sec:exceptional5}

As shown in \Cref{tab:arity5-experiment}, the only symmetric Boolean predicates of arity $5$ not resolved by our current lower- and upper-bound \critrs are
$\wt{R_{023}}=\{0,2,3\}$ and $\wt{R_{124}}=\{1,2,4\}$. For both predicates, the \OR-reduction lower-bound \critr certifies only an $\Omega(n^2)$ lower bound, while the $t$-balancedness upper-bound \critr gives an $O(n^3)$ upper bound. Thus these two predicates mark the first point in our classification where the two \critrs fail to meet. Resolving them therefore requires going beyond at least one of the two techniques used in the rest of the paper. In this subsection, we describe our attempts to do so and explain how the remaining obstruction leads naturally to an extremal set-system problem. As proved in \Cref{thm:positive-nrd}, we work with simple positive instances up to a constant-factor change in the number of variables. %Since the predicates here are symmetric, the order of the five variables in a simple constraint is immaterial. Thus we identify each simple positive arity-$5$ constraint with a $5$-element set $A\in\binom{[n]}{5}$.

\paragraph{Conditional \NRD.}
It is useful to compare ordinary non-redundancy with a conditional version of the problem, as studied in \cite{brakensiek2025redundancy,brakensiek2025richness}. Let $R\subseteq S\subseteq D^r$ be relations over domain $D$. An $R$-instance $I$ is \emph{$S$-conditionally non-redundant} if for every constraint $C=\{R,(\bar x)\}\in I$, there exists an assignment $\sigma$ such that $\sigma$ satisfies every constraint in $I\setminus\{C\}$, while the tuple induced by $\bar x$ under $\sigma$ lies in $S\setminus R$. Equivalently, $\sigma$ satisfies all other $R$-constraints and satisfies the relaxed constraint $\{S,(\bar x)\}$, but falsifies $C$. We write $\nrdx{R\mid S}$ for the maximum number of constraints in an $S$-conditionally non-redundant $R$-instance on $n$ variables.

The quantity $\nrdx{R\mid S}$ is useful because of the triangle inequality proved by \cite{brakensiek2025redundancy}:
\[
    \nrdx{R}\le \nrdx{R\mid S}+\nrdx{S}.
\]
For the predicate $R_{124}$, let $R_{1245}$ be the arity-$5$ symmetric predicate with $\wt{R_{1245}}=\{1,2,4,5\}$.
Then
\[
    \nrdx{R_{124}\mid R_{1245}}
    \le
    \nrdx{R_{124}}
    \le
    \nrdx{R_{124}\mid R_{1245}}+\nrdx{R_{1245}}.
\]
Our upper-bound framework gives $\nrdx{R_{1245}}=O(n^2)$; see \Cref{tab:arity5-experiment}. Hence, once one has the quadratic lower bound $\nrdx{R_{124}\mid R_{1245}}=\Omega(n^2)$, the original $\nrdx{R_{124}}$ and the conditional \nrdx{R_{124}\mid R_{1245}} agree up to constant factors whenever the conditional quantity is superquadratic. Thus, for understanding whether the true growth is closer to $n^2$ or $n^3$, the conditional problem is equivalent to the original problem.

\paragraph{Pairwise-intersection construction.}
A first attempt to lower-bound $\nrdx{R_{124}\mid R_{1245}}$ is the following pairwise-intersection construction. Let $\mathcal F\subseteq \binom{[n]}{5}$ be a family such that
\[
    |A\cap B|\in \{1,2,4\}
    \qquad\text{for all distinct } A,B\in\mathcal F.
\]
Such a family gives an $R_{1245}$-conditionally non-redundant simple positive $\CSP(R_{124})$ instance: introduce one constraint for each $A\in\mathcal F$. To witness that the constraint indexed by $A$ is essential, take the assignment that sets exactly the variables of $A$ to $1$ and all other variables to $0$. The constraint $A$ itself has value $|A\cap A|=5$, which lies in $\wt{R_{1245}}\setminus \wt{R_{124}}$, while every other constraint $B\neq A$ has value $|A\cap B|\in \{1,2,4\}$ and is therefore satisfied as an $R_{124}$-constraint.

This construction gives the desired quadratic lower bound. However, it cannot yield a superquadratic lower bound. As remarked in \cite{brakensiek2026classification}, a classical theorem of Deza, Erdős, and Frankl \cite{deza1978intersection} on uniform set systems with restricted intersection sizes implies that every such pairwise-intersection family has size $O(n^2)$; matching quadratic examples give $\Theta(n^2)$. Thus the pairwise-intersection approach explains the current $\Omega(n^2)$ lower bound but cannot close the gap to the $O(n^3)$ upper bound.

\paragraph{Generalized witness-set problem.}
The natural next step is to allow more general witnesses. For $W\subseteq\{0,1,\dots,5\}$, define $g_W(n)$ to be the maximum size of a family $\mathcal F\subseteq \binom{[n]}{5}$ such that for every $A\in\mathcal F$ there exists a set $X_A\subseteq[n]$ satisfying
\[
    |A\cap X_A|\notin W
    \qquad\text{and}\qquad
    |B\cap X_A|\in W
    \quad\text{for every }B\in\mathcal F\setminus\{A\}.
\]
The pairwise-intersection construction is the special case $X_A=A$ for every $A$. More importantly, for simple positive instances, where each constraint uses five distinct variables, this generalized witness-set problem is exactly the positive non-redundancy problem for the symmetric predicate $W$. Indeed, a family $\mathcal F$ gives an instance with one constraint for each $A\in\mathcal F$, and $X_A$ is precisely the assignment witnessing that the constraint $A$ is essential. Conversely, any simple positive non-redundant instance of a symmetric $5$-ary predicate gives such a family and such witnesses. 

For the conditional problem, the corresponding extremal quantity has the same form, except that the bad value for the distinguished constraint must lie in the relaxed set $S\setminus R$. Thus, for $\nrdx{R_{124}\mid R_{1245}}$, the condition becomes
\[
    |A\cap X_A|=5
    \qquad\text{and}\qquad
    |B\cap X_A|\in \{1,2,4\}
    \quad\text{for every }B\neq A.
\]
Thus, resolving $R_{124}$ is exactly\footnote{Recall that \nrd of the simple positive model is equivalent, up to constant factors, to the most general model; see \Cref{thm:positive-nrd}} the generalized witness-set problem above. The same discussion applies to $R_{023}$ by taking $R_{0235}$ with $\wt{R_{0235}}=\{0,2,3,5\}$; the corresponding conditional witness condition is $|A\cap X_A|=5$ and $|B\cap X_A|\in\{0,2,3\}$ for every $B\neq A$.

\section{Upper Bound}\label{sec:upper-bound}

In this section we prove the upper bound on the non-redundancy of \tbalanced{} Boolean relations. The proof has two ingredients.

First, in \Cref{subsec:t-balancedness}, we define \tbalancedness{} via a degree-$t$ monomial lift and prove that it is equivalent to the existence of degree-$\le t$ multilinear capturing polynomials over rings, see \Cref{thm:t-balanced-main}. The proof reduces the degree-$t$ statement to the degree-$1$ balancedness of \cite{chen2020best} in the lifted space. The key intermediate step is a relative version of their linear-capture theorem, stated as \Cref{thm:linrel}.

Second, in \Cref{subsec:main-nrd-upper-bound}, we combine \Cref{thm:t-balanced-main} with the general polynomial sparsification theorem of Chen, Jansen, and Pieterse \cite{chen2020best}, stated as \Cref{thm:cjp-general-capture}, to obtain an $O(n^t)$ sparsification for $\cspr$ whenever $R$ is \tbalanced. The desired non-redundancy upper bound $\nrdr=O(n^t)$ then follows immediately, since a non-redundant instance cannot be equivalent to a proper subinstance, see \Cref{thm:balanced-nrd}.

\subsection{\texorpdfstring{$t$}{t}-balancedness and capturing polynomials}
\label{subsec:t-balancedness}

We begin by defining the relative notion of balancedness in a set.

\begin{definition}[Relative balancedness]
Let $A \subseteq X \subseteq \Bool^N$. We say that $A$ is \emph{balanced relative to $X$} if for every odd integer $m \ge 1$ and every tuple $a^{(1)},a^{(2)},\dots,a^{(m)} \in A$, it holds that if the coordinatewise alternating sum $u := a^{(1)} - a^{(2)} + a^{(3)} - \cdots - a^{(m-1)} + a^{(m)}$ belongs to $X$, then $u \in A$.
\end{definition}

We now define the degree-$t$ lift.

% \begin{definition}[$t$-lift and \tbalancedness]\label{def:tbal}
% Fix integers $r \ge 1$ and $t \ge 1$ and let
% the family of subsets of $[r]$ of size at most $t$ be denoted as
% $
% \mathcal{S}_{r,t} := \{ S \subseteq [r] : 1 \le |S| \le t\}.
% $. The degree-$t$ lift\ is the map \amatr{is this clear enough that ${0,1}^S$ is a hypercube of dim $|S|$, indexed by elements of $S$. defined in \Cref{sec:prelims}}
% \[ 
% \nu_t : \Bool^r \to \Bool^{\mathcal{S}_{r,t}}
% \qquad\text{defined by}\qquad
% \nu_t(x)_S := \prod_{i \in S} x_i.
% \]
% For a relation $R \subseteq \Bool^r$, let $X_t := \nu_t(\Bool^r),$ be the lifted Boolean Cube and $ R^{(t)} := \nu_t(R)$ be the lifted relation. We say that $R$ is \tbalanced{} if $R^{(t)}$ is balanced relative to $X_t$.
% \end{definition}

\begin{definition}[$t$-lift and \tbalancedness]\label{def:tbal}
Fix integers $r\ge 1$ and $t\ge 1$, and let
$\mathcal S_{r,t}:=\{S\subseteq [r]:1\le |S|\le t\}$. 
The degree-$t$ lift is the map
$\nu_t:\{0,1\}^r\to \{0,1\}^{\mathcal S_{r,t}}$
defined by $\nu_t(x)_S:=\prod_{i\in S}x_i$.

For a relation $R\subseteq\{0,1\}^r$, let
$X_t:=\nu_t(\{0,1\}^r)$ and $R^{(t)}:=\nu_t(R)$.
We say that $R$ is \emph{\tbalanced} if $R^{(t)}$ is balanced relative to $X_t$.
\end{definition}

Next, we formalize the polynomial notion used in the upper-bound \critr.

\begin{definition}[Capturing polynomial]
Let $R \subseteq \Bool^r$ and let $f \in \Bool^r \setminus R$.
A polynomial $p$ over a ring $E$ \emph{captures $f$ with respect to $R$} if
\[
p(y)=0 \quad \text{for all } y \in R,
\qquad\text{and}\qquad
p(f)\neq 0.
\]
Equivalently, we may write $p\ \captures_R\ f$.
\end{definition}

The next lemma reformulates relative balancedness as an integer-span condition. 

\begin{lemma}\label{lem:bal-latt}
Let $A \subseteq X \subseteq \Bool^N$, and define $\mathrm{ext}(v) := (1,v) \in \Z^{N+1}.$
Then the following are equivalent.
\begin{enumerate}
    \item $A$ is balanced relative to $X$.
    \item For every $u \in X$, if $\mathrm{ext}(u) \in \mathrm{span}_{\Z}\bigl(\{\mathrm{ext}(a): a \in A\}\bigr),$
    then $u \in A$.
\end{enumerate}
\end{lemma}

\begin{proof}
Let's first prove the forward direction, i.e., $(1)\Rightarrow (2)$. Assume that $A$ is balanced relative to $X$, and let $u \in X$ satisfy $\mathrm{ext}(u)=\sum_{j=1}^s \gamma_j \,\mathrm{ext}(a^{(j)})$ for some $a^{(1)},\dots,a^{(s)} \in A$ and $\gamma_1,\dots,\gamma_s \in \Z$.
Looking at the first coordinate gives $1=\sum_{j=1}^s \gamma_j$.
Now expand this integer combination as an odd alternating sum. I.e., for each $\gamma_j>0$,
take $\gamma_j$ copies of $a^{(j)}$ with sign $+$, and for each $\gamma_j<0$, take
$-\gamma_j$ copies of $a^{(j)}$ with sign $-$. Since the total sum of coefficients is $1$,
the number of positively signed copies is exactly one more than the number of negatively
signed copies. After reordering, we have
\[
\mathrm{ext}(u)
=
\mathrm{ext}(b^{(1)})-\mathrm{ext}(b^{(2)})+\mathrm{ext}(b^{(3)})-\cdots+\mathrm{ext}(b^{(m)})
\]
for some odd $m$ and some $b^{(1)},\dots,b^{(m)} \in A$. This implies $u=b^{(1)}-b^{(2)}+b^{(3)}-\cdots+b^{(m)}.$ Since $u \in X$ and $A$ is balanced relative to $X$, it follows that $u \in A$.

For the reverse direction, let $u=a^{(1)}-a^{(2)}+a^{(3)}-\cdots+a^{(m)}$ for some odd $m$ and some $a^{(1)},\dots,a^{(m)}\in A$, and assume $u \in X$.
Then
\[
\mathrm{ext}(u)
=
\mathrm{ext}(a^{(1)})-\mathrm{ext}(a^{(2)})+\mathrm{ext}(a^{(3)})-\cdots+\mathrm{ext}(a^{(m)}),
\]
since the first coordinate on the right-hand side is $1$.
Hence $\mathrm{ext}(u)$ belongs to the integer span of $\{\mathrm{ext}(a):a\in A\}$,
and so by assumption $u \in A$. Therefore $A$ is balanced relative to $X$.
\end{proof}

We now prove the analogue, in the set $X$, of the balancedness theorem of \cite{chen2020best} for ordinary Boolean relations.

\begin{lemma}
\label{thm:linrel}
Let $A \subseteq X \subseteq \Bool^N$. Then the following are equivalent.
\begin{enumerate}
    \item $A$ is balanced relative to $X$.
    \item For every $u \in X \setminus A$, there exists a prime-power $q_u := p^a$ (for some prime $p$) and a linear polynomial
    $\ell_u(z_1,\dots,z_N)=\alpha_0+\alpha_1 z_1+\cdots+\alpha_N z_N$
    over $\Z/q_u\Z$ such that $\ell_u\ \captures_A\ u$. %$\ell_u(a)=0$ for all $a \in A,$ and $\ell_u(u)\neq 0$.
\end{enumerate}
Moreover, if $A$ is not balanced relative to $X$, then there exists some
$u \in X \setminus A$ that is not capturable by any linear polynomial over any \ring.
\end{lemma}

\begin{proof}
Let's first prove $(1)\Rightarrow(2)$. Assume that $A$ is balanced relative to $X$, and
fix $u \in X \setminus A$. By \Cref{lem:bal-latt}, we have that $\mathrm{ext}(u)\notin \mathrm{span}_{\Z}\bigl(\{\mathrm{ext}(a):a\in A\}\bigr)$.
Let $M$ be the matrix whose rows are the vectors $\mathrm{ext}(a)$ for $a\in A$,
followed by $\mathrm{ext}(u)$ as the final row.

Suppose for contradiction that no such linear capturing polynomial exists.
Then for every prime-power $q$ and every nonzero
$c \in \Z/q\Z$, the linear system
\[
M \alpha = b_c,
\qquad
b_c := (0,\dots,0,c)^T
\]
has no solution over $\Z/q\Z$.
By \cite[Lemma~4.6]{chen2020best}, this implies that for every prime-power $q$,
the last row $\mathrm{ext}(u)$ lies in the span of the previous rows modulo $q$.
Then by \cite[Lemma~4.3]{chen2020best}, it follows that $\mathrm{ext}(u)\in \mathrm{span}_{\Z}\bigl(\{\mathrm{ext}(a):a\in A\}\bigr),$
contradicting \Cref{lem:bal-latt}.
Therefore there exists a prime-power $q$ and coefficients $\alpha_0,\alpha_1,\dots,\alpha_N \in \Z/q\Z$
such that
\[
\alpha_0+\alpha_1 a_1+\cdots+\alpha_N a_N = 0
\qquad\text{for all } a\in A,
\]
while
\[
\alpha_0+\alpha_1 u_1+\cdots+\alpha_N u_N \neq 0
\qquad\text{in } \Z/q\Z.
\]
Thus
\[
\ell_u(z_1,\dots,z_N):=\alpha_0+\alpha_1 z_1+\cdots+\alpha_N z_N
\]
captures $u$ with respect to $A$.

Next let's prove the contrapositive of $(2)\Rightarrow(1)$, in the stronger form stated at the end of the theorem. 
Assume that $A$ is not balanced relative to $X$. Then there
exist an odd integer $m \ge 1$ and tuples $a^{(1)},\dots,a^{(m)} \in A$ such that
\[
u:=a^{(1)}-a^{(2)}+a^{(3)}-\cdots+a^{(m)} \in X \setminus A.
\]
We claim that no linear polynomial over any \ring captures $u$ with respect to $A$.

Suppose for contradiction that some \ringg{E} and some linear polynomial
\[
\ell(z_1,\dots,z_N)=\beta_0+\beta_1 z_1+\cdots+\beta_N z_N
\]
over $E$ capture $u$. By our ring convention, Boolean coordinates and the signs $\pm 1$ are
interpreted in $E$ via the map $\mathbb Z\to E$, $k\mapsto k\cdot 1_E$.

Then $\ell(a)=0$ for all $a\in A$, and $\ell(u)\neq 0$.
But our assumption gives
\[
\ell(u)
=
\beta_0+\sum_{i=1}^N \beta_i
\bigl(a^{(1)}_i-a^{(2)}_i+a^{(3)}_i-\cdots+a^{(m)}_i\bigr).
\]
Since $m$ is odd, the constant term also alternates correctly, and hence
\[
\ell(u)=\ell(a^{(1)})-\ell(a^{(2)})+\ell(a^{(3)})-\cdots+\ell(a^{(m)}).
\]
Since each $a^{(j)}$ lies in $A$, the right-hand side is $0$, we get  a contradiction.
Hence $u$ is not capturable over any \ring.
\end{proof}

We now obtain the desired equivalence between \tbalancedness{} and degree-$\le t$ capturing.

\begin{theorem}
\label{thm:t-balanced-main}
Let $R\subseteq\{0,1\}^r$ and let $t\ge 1$. Then the following are equivalent.
\begin{enumerate}
    \item $R$ is \tbalanced.
    \item For every $f\in\{0,1\}^r\setminus R$, there exists a prime-power $q_f$ and a multilinear polynomial $p_f$ over $\mathbb Z/q_f\mathbb Z$ of degree at most $t$ such that $p_f$ captures $f$ with respect to $R$.
\end{enumerate}
Moreover, if $R$ is not \tbalanced, then there exists $f\in\{0,1\}^r\setminus R$ that is not captured by any degree-$\le t$ multilinear polynomial over any \ring.
\end{theorem}

\begin{proof}
By definition, $R$ is \tbalanced{} if and only if $R^{(t)}$ is balanced relative to $X_t$.

We first prove $(1)\Rightarrow(2)$. Assume that $R$ is \tbalanced{} and fix $f \in \Bool^r \setminus R$. Since $t \ge 1$, we have that the lift $\nu := \nu_t(f)$ lies in $X_t \setminus R^{(t)}$ (because it clearly is in $X_t$, and it cannot be in $R^{(t)}$ since $f\notin R$).
%all singleton monomials occur among the coordinates of $\nu_t$, and therefore $\nu_t$ is injective. In particular,$\nu := \nu_t(f)$ lies in $X_t \setminus R^{(t)}$.

Applying \Cref{thm:linrel} to the pair
$R^{(t)} \subseteq X_t$, there exist a prime-power $q_f$ and a linear polynomial $\ell_\nu(z)=\alpha_0+\sum_{S \in \mathcal{S}_{r,t}} \alpha_S z_S$ over $\Z/q_f\Z$ such that $\ell_\nu\ \captures\ \nu$ with respect to $R^{(t)}$, i.e., $\ell_\nu(v)=0,$ for all $v \in R^{(t)},$ and $\ell_\nu(\nu)\neq 0$.
Define
\[
p_f(x_1,\dots,x_r)
:=
\ell_\nu(\nu_t(x_1,\dots,x_r))
=
\alpha_0+\sum_{S \in \mathcal{S}_{r,t}} \alpha_S \prod_{i \in S} x_i.
\]
This is a multilinear polynomial of degree at most $t$ over $\Z/q_f\Z$.
If $y \in R$, then $\nu_t(y) \in R^{(t)}$, so
\[
p_f(y)=\ell_\nu(\nu_t(y))=0.
\]
On the other hand,
\[
p_f(f)=\ell_\nu(\nu_t(f))=\ell_\nu(\nu)\neq 0.
\]
Therefore $p_f$ $\captures_R$ $f$.

Now we prove $(2)\Rightarrow(1)$, in the stronger form stated at the end of the theorem. Assume that for every $f \in \Bool^r \setminus R$,
there exists a degree-$\le t$ multilinear capturing polynomial over some ring. We show that $R^{(t)}$ is balanced relative to $X_t$. Suppose not. Then, by \Cref{thm:linrel}, there exists $\nu \in X_t \setminus R^{(t)}$ that is not capturable by any linear polynomial over any \ring. Since the $t$-lift $\nu_t$ is injective, there is a unique $f \in \Bool^r \setminus R$ such that $\nu=\nu_t(f)$. By assumption, there exists a \ringg{E} and a degree-$\le t$ multilinear
polynomial
\[
p_f(x_1,\ldots,x_r)=\beta_0+\sum_{S\in S_{r,t}}\beta_S\prod_{i\in S}x_i
\]
over $E$ that captures $f$ with respect to $R$.
Define a linear polynomial on the lifted coordinates by
\[
\ell_\nu(z):=\beta_0+\sum_{S\in\mathcal{S}_{r,t}} \beta_S z_S.
\]
If $\nu'=\nu_t(y)\in R^{(t)}$, then $y \in R$, and therefore
\[
\ell_\nu(\nu')
=
\beta_0+\sum_{S\in\mathcal{S}_{r,t}} \beta_S \nu_t(y)_S
=
\beta_0+\sum_{S\in\mathcal{S}_{r,t}} \beta_S \prod_{i\in S} y_i
=
p_f(y)
=
0.
\]
But
\[
\ell_\nu(\nu)=\ell_\nu(\nu_t(f))=p_f(f)\neq 0.
\]
So $\ell_\nu$ captures $\nu$ with respect to $R^{(t)}$, contradicting the choice of $\nu$.
Hence $R^{(t)}$ is balanced relative to $X_t$, that is, $R$ is \tbalanced.
\end{proof}

\subsection{Main \texorpdfstring{$\NRD$}{NRD} upper bound proof}
\label{subsec:main-nrd-upper-bound}

We now combine \Cref{thm:t-balanced-main} with the general polynomial
sparsification theorem of \cite{chen2020best} that gives a polynomial-time algorithm for finding a smaller \emph{equivalent} instance (i.e., subinstance with the same set of satisfying assignments) given a polynomial representation.

\begin{theorem}[Theorem 3.5, \cite{chen2020best}]\label{thm:cjp-general-capture}
Let $R \subseteq \Bool^r$ be a fixed Boolean relation. Suppose that for every $u \in \Bool^r \setminus R$, there exists a ring $E_u$ which is either $\mathbb{Q}$ or a $\Z/q_u\Z$ for some $q_u\in \mathbb{N}_{>1}$ and a polynomial $p_u$ over $E_u$ of degree at most $t$ such that $p_u$ captures $u$ with respect to $R$. Then there exists a polynomial-time algorithm that, given a set of $R$-constraints on $n$ variables, outputs an equivalent subset of $O(n^t)$ constraints.
\end{theorem}

% \begin{proof}
% This is the upper-bound theorem of \cite[Theorem~3.5]{chen2020best}
% specialized to a single Boolean relation.
% \end{proof}

\begin{corollary}\label{cor:t-balanced-sparsification}
Let $R \subseteq \Bool^r$ be \tbalanced, with $t \ge 1$. Then every instance
of $\cspr$ on $n$ variables has an equivalent subinstance with $O(n^t)$
constraints.
\end{corollary}

\begin{proof}
By \Cref{thm:t-balanced-main}, for every tuple
$f \in \Bool^r \setminus R$, there exists a prime-power $q_f$ and a multilinear
polynomial $p_f(x_1,\dots,x_r)$ over $\Z/q_f\Z$
of degree at most $t$ that $\captures_R$ $f$.
Hence, by \Cref{thm:cjp-general-capture}, every $n$-variable instance of $\cspr$ has an equivalent subinstance of
$O(n^t)$ constraints.
\end{proof}

\begin{theorem}\label{thm:balanced-nrd}
Let $R \subseteq \Bool^r$ be \tbalanced, with $t \ge 1$. Then $\nrdr=O(n^t).$
\end{theorem}

\begin{proof}
Consider any non-redundant instance $\mathcal{I}$ of $\cspr$ on $n$ variables. By \Cref{cor:t-balanced-sparsification}, $\mathcal{I}$ has an equivalent
subinstance $\mathcal{I}'$ with $O(n^t)$ constraints. Since $\mathcal{I}$ is non-redundant, no proper subset of its constraints is
equivalent to $\mathcal{I}$. Hence $\mathcal{I}'=\mathcal{I}$, and therefore
$\mathcal{I}$ itself has $O(n^t)$ constraints. This proves the claim.
\end{proof}

Consequently, if $u(R):=\min\{t\ge 1 : R \text{ is } \tbalanced\}$, then $\nr{R}{n}=O(n^{u(R)})$.

\section{Lower Bound}\label{sec:lower-bound}

% In this section we recall the existing lower-bound framework of Carbonnel~\cite{carbonnel2022redundancy}, which already gives the lower bounds we need. We do not introduce any new lower-bound technique here. Instead, we restate the relevant notions and techniques. This provides the algebraic \critr that we use in our experiments and later classification results.

In this section we recall the lower-bound framework of Carbonnel~\cite{carbonnel2022redundancy}. The role of this section is to justify the lower-bound test used in \Cref{sec:experiments}. The framework has two ingredients. First, fgpp-definability formalizes when constraints of one relation can be replaced by constant-size gadgets over another relation while preserving non-redundancy lower bounds. Second, pattern partial polymorphisms give an algebraic characterization of fgpp-definability. For our purposes, the only pattern needed is Carbonnel's universal $k$-cube pattern: if this pattern fails to preserve $R$, then $R$ inherits the $\Omega(n^k)$ non-redundancy lower bound of $\OR_k$.

% \amati{i might have written it in a confusing and a hand-wavy way, check} At a high level, Carbonnel's framework gives an algebraic way to certify a lower bound on the
% non-redundancy. The first key notion they define is \emph{fgpp-definitions} which
% describes when one relation can be expressed from another in a way that preserves
% non-redundancy lower bounds. Secondly, they define \emph{pattern partial polymorphisms} which
% give an algebraic obstruction to such definability. The key point for us is that a
% particular pattern, the universal $k$-cube, witnesses the ability of any relation to define the
% $k$-ary \OR relation. Since $\OR_k$ has non-redundancy $\Omega(n^k)$, failure of this
% pattern immediately yields an $\Omega(n^k)$ lower bound for $R$.

% We now recall these relevant notions now.

Before defining fgpp-definability formally, let us explain their role. In ordinary pp-definability (famously used by Schaefer \cite{Schaefer78} for \textsf{P} vs \textsf{NP} classification of problems), one expresses a target relation by a conjunction of constraints over a base language. Fgpp-definability extends this by allowing certain unary functional guards required for non-redundancy lower-bounds.

\begin{definition}[fgpp-definition]\label{def:fgppdef}
Let $\Gamma$ be a constraint language over a finite domain $D$, and let $R \subseteq D^m$
be a relation. We say that $R$ is \emph{fgpp-definable} from $\Gamma$ if there is a formula
\[
R(x_1,\dots,x_m)\equiv \exists y_1,\dots,y_q\; \psi(x_1,\dots,x_m,y_1,\dots,y_q),
\]
where $\psi$ is a conjunction of constraints over $\Gamma$ together with \emph{functional guards}
of the form $Q_g(u,v)$, where $g\colon D \to D$ is a unary map and
\[
Q_g:=\{(a,g(a)):a\in D\}.
\]
We write $\langle \Gamma \rangle_{fg}$ for the set of all relations fgpp-definable from $\Gamma$.
\end{definition}

Next, we define the pattern partial polymorphisms. The intuition is that a
\emph{pattern} specifies a partial operation by listing a collection of allowed input-output
templates. Such operations are more structured than arbitrary partial polymorphisms, but still
rich enough to capture the definability questions relevant here.

\begin{definition}[Polymorphism pattern and pattern partial operation]
A \emph{polymorphism pattern} of arity $s$ is a set of pairs
\[
((x_1,\dots,x_s),x),
\]
where $(x_1,\dots,x_s)$ is a tuple of variables and $x$ is one of the variables appearing in it.
Given a finite domain $D$, the \emph{interpretation} of such a pattern $P$ on $D$ is the
partial operation $f_P^D$ defined as follows:
its domain consists of all tuples of the form
\[
(\varphi(x_1),\dots,\varphi(x_s))
\]
arising from a pair $((x_1,\dots,x_s),x)\in P$ and an assignment $\varphi$ of the variables to $D$,
and on such a tuple we set
\[
f_P^D(\varphi(x_1),\dots,\varphi(x_s)):=\varphi(x).
\]
A partial operation obtained in this way is called a \emph{pattern partial operation}.
\end{definition}
We only consider patterns for which this rule is well-defined.

Next we recall what it means for such an operation to preserve a relation. This is the usual
coordinatewise preservation notion for partial polymorphisms, restricted to the operations that
come from patterns.

\begin{definition}[Pattern partial polymorphisms]
Let $R \subseteq D^r$ be a relation and let $f$ be a $k$-ary partial operation on $D$.
We say that $f$ is a \emph{partial polymorphism} of $R$ if for every choice of tuples
$t^{(1)},\dots,t^{(k)} \in R$ such that $f$ is defined coordinatewise on these tuples,
the tuple obtained by applying $f$ coordinatewise also belongs to $R$.
If in addition $f$ is a pattern partial operation, then we call it a
\emph{pattern partial polymorphism} of $R$.

For a constraint language $\Gamma$, we write $\pppol(\Gamma)$ for the set of all pattern partial polymorphisms of every relation in $\Gamma$.
\end{definition}

The importance of $\pppol(\Gamma)$ is that it gives an exact characterization of fgpp-definability: definability on one side corresponds to inclusion of pattern partial
polymorphisms on the other.

\begin{theorem}[Carbonnel~{\cite[Proposition~7]{carbonnel2022redundancy}}]
\label{thm:carbonnel-duality}
Let $\Gamma_1,\Gamma_2$ be finite constraint languages over the same finite domain. Then
\[
\pppol(\Gamma_1)\subseteq \pppol(\Gamma_2)
\qquad\Longleftrightarrow\qquad
\Gamma_2 \subseteq \langle \Gamma_1\rangle_{fg}.
\]
In particular, for a relation $R$ and any relation $S$ over the same domain,
\[
S \in \langle \{R\}\rangle_{fg}
\qquad\Longleftrightarrow\qquad
\pppol(\{R\}) \subseteq \pppol(\{S\}).
\]
\end{theorem}

To turn this into concrete lower bounds, Carbonnel defines one specific pattern, the
universal $k$-cube pattern. Its role is simple: failure of this pattern is exactly the
\critr that $R$ can encode the $k$-ary \OR relation, which is the canonical source of
an $\Omega(n^k)$ lower bound.

\begin{definition}[Universal $k$-cube pattern~{\cite[Definition~13]{carbonnel2022redundancy}}]
\label{def:universal-k-cube}
Fix $k \ge 2$. Let $c_1,\dots,c_{2^k-1}$ be the lexicographic ordering of the tuples in $\{x,y\}^k \setminus \{(y,\dots,y)\}$
with respect to $y>x$. The \emph{universal $k$-cube pattern} $P_k^u$ of arity $2^k-1$ is the set of pairs $(t_i,y)$ for all $ 1\le i\le k$, where
\[
t_i := \bigl(c_1[i],c_2[i],\dots,c_{2^k-1}[i]\bigr).
\]
Its interpretation on the Boolean domain is denoted by $f^{\Bool}_{P_k^u}$.
\end{definition}

Under the identification $x\mapsto 1$ and $y\mapsto 0$, the Boolean interpretation
$f^{\{0,1\}}_{P^u_k}$ is exactly the universal $k$-cube operation $f_k$ used in
\Cref{sec:experiments}: its domain consists of $0^m,1^m,c_i,1-c_i$ for
$i\in[k]$, where $m=2^k-1$, and its values are
$f_k(0^m)=0$, $f_k(1^m)=1$, $f_k(c_i)=0$, and $f_k(1-c_i)=1$.

For example, when $k=3$, the left-hand sides of the pairs in $P_3^u$ are exactly the
three columns of the $\OR_3$ relation, and the right-hand side is the missing
tuple $(0,0,0)$, after the identification $x\mapsto 1$ and $y\mapsto 0$~\cite[Example~14]{carbonnel2022redundancy}.

We can now state the lower-bound theorem we need. It says that if the universal $k$-cube
pattern fails on $R$, then $R$ inherits the $\Omega(n^k)$ non-redundancy lower bound of
$\OR_k$.

% \begin{theorem}[Carbonnel~{\cite[Lemma~15]{carbonnel2022redundancy}}]
% \label{thm:carbonnel-lower-bound}
% Let $R$ be a relation of arity $r$ over a finite domain $D$, and let $2 \le k \le r$.
% If
% \[
% f^D_{P_k^u} \notin \pppol(\{R\}) \implies
% \nrdr=\Omega(n^k).
% \]
% In particular, for Boolean relations,
% \[
% f^{\Bool}_{P_k^u} \notin \pppol(\{R\}) \implies
% \nrdr=\Omega(n^k).
% \]
% \end{theorem}

\begin{theorem}[{\cite[Lemma~15]{carbonnel2022redundancy}}]\label{thm:carbonnel-lower-bound}
Let $R$ be a relation of arity $r$ over a finite domain $D$, and let $2\le k\le r$. If
$f^D_{P^u_k}\notin \pppol(\{R\})$, then $\nrdr=\Omega(n^k)$.
\end{theorem}

For us, this theorem provides a direct computational test: if we can certify that the
universal $k$-cube pattern does not preserve $R$, then we immediately obtain an
$\Omega(n^k)$ lower bound on $\nrdr$.
Consequently,
if 
$
\ell(R):=\max\Bigl(\{1\}\cup
\{k\ge 2 : f^{\{0,1\}}_{P^u_k}\notin \pppol(\{R\})\}\Bigr).
$
then $\nrdr=\Omega\bigl(n^{\,\ell(R)}\bigr)$.

\section{Literals, Constants, and Repetitions}
\label{sec:positive-nrd}

%\amati{Proved the model without repitions, constants, and literals is equivalent to the one with.}

In this section we justify that allowing literals, constants, repeated variables, or only simple positive constraints changes non-redundancy only up to a constant-factor blow-up in the number of variables. The issue is that fgpp-definability naturally allow unary substitutions such as literals and constants, whereas the set-system discussion in \Cref{sec:exceptional5} uses simple positive constraints, i.e., constraints on distinct variables. We show that, for fixed finite domains and fixed arity, these models are equivalent up to a constant-factor blow-up in the number of variables. The proof is by a simple variable-splitting construction: for each original variable, we create copies indexed by the unary substitution and by the coordinate position in the constraint. This removes literals, constants, and repeated variables while preserving non-redundancy witnesses.

Let $D$ be a fixed finite domain, let $\Gamma$ be a fixed finite constraint language containing $r$-ary of a relation over $D$. We use three variants of non-redundancy in this section.
\begin{enumerate}
    \item First, $\nrp{\Gamma}{n}{+}$ denotes the usual positive tuple model: constraints have the form $\{R,(x_{i_1},\ldots,x_{i_k})\}$, where $R\in\Gamma$ has arity $k$, and the tuple $(x_{i_1},\ldots,x_{i_k})$ is ordered and repetitions are allowed.
    \item Second, $\nrp{\Gamma}{n}{\mathrm{simp}}$ denotes the simple positive model: constraints again have the form $\{R,(x_{i_1},\ldots,x_{i_k})\}$, and the variables in the tuple are requred to be distinct.
    \item Third, $\nrp{\Gamma}{n}{\mathrm{lit}}$ denotes the literal model: constraints have the form $\{R,(g_1(x_{i_1}),\ldots,g_k(x_{i_k}))\}$, where $R\in\Gamma$ and each $g_j$ belongs to a fixed finite set $\mathcal U$ of unary maps $D\to D$ containing the identity map.
\end{enumerate}

We treat instances as sets of constraints. Repeated identical constraints are irrelevant for non-redundancy, since duplicate copies are automatically redundant.

\begin{theorem}
\label{thm:positive-nrd}
Let $D$ be a domain of size $|D|=q$, let $\Gamma$ be a fixed finite constraint language containing $r$-ary relations over $D$. Let $\mathcal U$ be a finite set of unary maps $D\to D$ containing the identity map. Then, for every $n\ge 1$, we have
$$\nrp{\Gamma}{n}{\mathrm{simp}}\le \nrp{\Gamma}{n}{+}\le \nrp{\Gamma}{n}{\mathrm{lit}}\le \nrp{\Gamma}{O_{r,q}(n)}{\mathrm{simp}}.$$
Consequently, for fixed $D,\Gamma$, and $\mathcal U$, the three models have the same asymptotic non-redundancy up to constant-factor changes in the number of variables.
\end{theorem}

\begin{proof}
The first two inequalities are immediate. Every simple positive instance is a positive tuple instance, and every positive tuple instance is a literal model instance because $\mathcal U$ contains the identity map.

It remains to prove $\nrp{\Gamma}{n}{\mathrm{lit}}\le \nrp{\Gamma}{O_{r,q}(n)}{\mathrm{simp}}$. Let $I$ be a non-redundant literal model instance over variables $x_1,\ldots,x_n$. Thus every constraint $C$ of $I$ has the form $C=\{R,(g_1(x_{i_1}),\ldots,g_r(x_{i_r}))\}$, where $R\in\Gamma$ and $g_1,\ldots,g_r\in\mathcal U$.

We construct a simple positive $\cspg$ instance $I^{\mathrm{simp}}$. Its variables are $x_i^{g,j}$ for $i\in[n]$, $g\in\mathcal U$, and $j\in[r]$. Upper bounding the size of $\mathcal U$ by the set of all unary maps $D\to D$, gives $|\mathcal U|\leq |D|^{|D|}$. Thus $I^{\mathrm{simp}}$ has $r\cdot|\mathcal U|\cdot n = O_{r,q}(n)$ variables. We translate each constraint $C=\{R,(g_1(x_{i_1}),\ldots,g_r(x_{i_r}))\}$ to the positive constraint
$C^{\mathrm{simp}}=\{R,(x_{i_1}^{g_1,1},x_{i_2}^{g_2,2},\ldots,x_{i_r}^{g_r,r})\}$.
This constraint is simple: even if the original constraint repeats a variable or a unary map, the coordinate index $j$ is different in each position. The translated instance has the same number of constraints as $I$.

We now show that $I^{\mathrm{simp}}$ is non-redundant. Fix a constraint $C$ of $I$. Since $I$ is non-redundant, there is an assignment $\alpha:\{x_1,\ldots,x_n\}\to D$ that satisfies every constraint of $I$ except $C$, and falsifies $C$. Define an assignment $\beta$ to the variables of $I^{\mathrm{simp}}$ by $\beta(x_i^{g,j})=g(\alpha(x_i))$.

By construction, for every constraint $C'$ of $I$, the translated constraint $(C')^{\mathrm{simp}}$ has under $\beta$ exactly the same truth value that $C'$ has under $\alpha$. Therefore $\beta$ satisfies every translated constraint except $C^{\mathrm{simp}}$, and falsifies $C^{\mathrm{simp}}$. Hence $C^{\mathrm{simp}}$ is not implied by the other translated constraints.

Since this holds for every constraint $C$ of $I$, the translated instance $I^{\mathrm{simp}}$ is non-redundant. Therefore every non-redundant literal model instance on $n$ variables yields a non-redundant simple positive instance on $O_{r,q}(n)$ variables with the same number of constraints. This proves the final inequality.
\end{proof}

% Taking $\mathcal U=\mathcal U_D$, the set of all unary maps $D\to D$, gives $|\mathcal U_D|=|D|^{|D|}$. Thus, for every fixed finite domain and fixed finite language, allowing arbitrary unary substitutions, including constants, changes $\NRD$ only by a constant-factor blow-up in the number of variables.

In the Boolean case, if we only need literals and constants, it suffices to take $\mathcal U=\{\mathrm{id},\neg,0,1\}$, so $|\mathcal U|=4$. Hence the blow-up from the literal-and-constant model to the simple positive model is at most $4rn$ variables. In light of this theorem, throughout the paper we suppress these model
distinctions when discussing asymptotic non-redundancy.

\section{Conclusion}\label{sec:conclusion}

We studied the non-redundancy of symmetric Boolean predicates. On the upper-bound side, we introduced \tbalancedness{}, a higher-degree extension of the balancedness framework of Chen, Jansen, and Pieterse~\cite{chen2020best}, and showed that it yields polynomial capturing representations and hence upper bounds on $\nrdr$. On the lower-bound side, we used Carbonnel's~\cite{carbonnel2022redundancy} universal-cube \critr to certify reductions from $\OR_k$ and obtain matching lower bounds. Combining these two \critrs, we classified the asymptotic growth of $\nrdr$ for every symmetric Boolean predicate of arity at most $4$, and for all but two predicates of arity $5$.

The most immediate open problem is to resolve the two remaining arity-$5$ predicates,
\[
    \wt{R_{023}}=\{0,2,3\}
    \qquad\text{and}\qquad
    \wt{R_{124}}=\{1,2,4\}.
\]
For both predicates, our current \critrs give only $\Omega(n^2) \le \nrdr \le O(n^3)$.
Closing this gap would complete the arity-$5$ classification. It would also clarify whether the lifted balancedness upper bound or the universal-cube lower bound can be strengthened, or whether these predicates require genuinely different techniques.

A second direction is to understand non-redundancy beyond symmetric Boolean predicates. Symmetry makes the problem amenable to a finite weight-set analysis, but this structure disappears for general languages. A concrete benchmark is the ternary non-Boolean predicate $\mathrm{BCK}:=\{111,222,012,120,201\}\subseteq\{0,1,2\}^3$ from Bessiere, Carbonnel, and Katsirelos~\cite{bessiere2020chain}, which has been highlighted in later work as a predicate whose non-redundancy is not explained by the standard universal-cube lower-bound machinery~\cite{brakensiek2025richness}. Understanding such examples may require new lower-bound mechanisms beyond reductions from $\OR_k$, and could be a step toward a structural characterization of the exponent governing $\nrdg$ for arbitrary finite constraint languages.

% In this paper we studied the non-redundancy of symmetric Boolean predicates. On the upper-bound side, we introduced a higher-degree extension of the balancedness framework of Chen, Jansen, and Pieterse \cite{chen2020best}; on the lower-bound side, we used Carbonnel's \cite{carbonnel2022redundancy} algebraic framework  based on reductions from OR. Combining these two ingredients, we obtained an exact classification of the growth of $\nrdr$ for every symmetric Boolean predicate of arity at most $4$, and an almost complete classification at arity $5$. At arity $5$, the only unresolved predicates are those corresponding to the weight sets $\{0,2,3\}$ and $\{1,2,4\}$.

% The most immediate open problem is therefore to determine the correct asymptotic growth of $\nrdr$ for these two predicates. At present, our framework yields an $\Omega(n^2)$ lower bound and an $O(n^3)$ upper bound for each of them. Resolving this gap would complete the arity-$5$ picture and, more importantly, clarify whether the current upper- and lower-bound methods can be strengthened to meet, or whether genuinely new ideas are needed.

% A broader direction is to understand non-redundancy beyond the symmetric Boolean setting. The symmetric case is already rich but still highly structured. For more general languages, earlier work suggests that the behavior of non-redundancy is more varied, and reductions from \OR\amati{maybe mention BCK here} are unlikely to capture the whole picture. Finding the right structural framework for asymmetric languages remains a natural next step.

\section*{Acknowledgement}

We thank Joshua Brakensiek, Venkatesan Guruswami, and Aaron Putterman for helpful discussions about non-redundancy of low-arity Boolean predicates, and in particular for pointing us to the Deza--Erdos--Frankl \cite{deza1978intersection} bound on uniform set systems with restricted intersections.

\paragraph{AI disclosure.}
GPT-5.4 Thinking was used to assist with parts of the computational implementation and with limited editorial polishing of the manuscript. The pseudocode, algorithmic design, proofs, mathematical statements, and overall structure of the paper were developed by the authors. The AI system was used only to help convert the authors' specifications into executable code and to assist with routine implementation, debugging, and writing refinement. The authors take full responsibility for the correctness of the code, results, and exposition.

\printbibliography

% \appendix

\end{document}